\documentclass{article}
\usepackage[utf8]{inputenc}
\usepackage{cite}
\usepackage{graphicx}
\usepackage[export]{adjustbox}
\usepackage{textcomp}
\usepackage{graphicx}
\usepackage{mathrsfs}
\usepackage{amsmath}
\usepackage{amsthm}
\usepackage{amssymb}
\usepackage{mathabx}
\usepackage{xspace}
\usepackage{hyperref}
\usepackage{fullpage}
\usepackage{listings}
\usepackage{tikz}

\usetikzlibrary{positioning,arrows.meta}
\usepackage[disable]{todonotes}

\usepackage{cleveref}

\newenvironment{psmallmatrix}{\left(\begin{smallmatrix}}{\end{smallmatrix}\right)}

\newtheorem*{remark}{Remark}
\newtheorem{theorem}{Theorem}[section]
\newtheorem*{theorem*}{Theorem}
\newtheorem{corollary}{Corollary}[theorem]
\newtheorem{lemma}[theorem]{Lemma}
\newtheorem{proposition}[theorem]{Proposition}

\theoremstyle{definition}
\newtheorem{definition}{Definition}[section]
\crefname{lemma}{lemma}{lemmas}
\Crefname{lemma}{Lemma}{Lemmas}
\crefname{proposition}{proposition}{propositions}
\Crefname{proposition}{Proposition}{Propositions}
\crefname{claim}{claim}{claims}
\Crefname{claim}{Claim}{Claims}
\Crefname{corollary}{corollary}{corollaries}
\Crefname{corollary}{Corollary}{Corollaries}
\usepackage{mathtools, stmaryrd}
\usepackage{xparse} 
\DeclarePairedDelimiterX{\Iintv}[1]{\llbracket}{\rrbracket}{\iintvargs{#1}}
\NewDocumentCommand{\iintvargs}{>{\SplitArgument{1}{,}}m}
{\iintvargsaux#1} %
\NewDocumentCommand{\iintvargsaux}{mm} {#1\mkern1.5mu..\mkern1.5mu#2}

\usepackage{array}
\newcolumntype{L}{>{$}p{.4\textwidth}<{$}} 
\newcolumntype{R}{>{$}p{.55\textwidth}<{$}}

\newcommand{\newclass}[2]{\newcommand{#1}{{\text{\upshape\sffamily #2}}\xspace}}

\newclass{\NP}{NP}
\renewcommand{\L}{{\text{\upshape\sffamily L}}\xspace}
\newclass{\NL}{NL}
\newclass{\RNC}{RNC}
\newclass{\CL}{CL}
\newclass{\CLP}{CLP}
\newclass{\TIME}{TIME}
\newclass{\SPACE}{SPACE}
\newclass{\CSPACE}{CSPACE}
\newclass{\NCo}{NC$^1$}
\newclass{\SACo}{SAC$^1$}
\newclass{\ACo}{AC$^1$}
\newclass{\TCo}{TC$^1$}
\newclass{\NCt}{NC$^2$}
\newclass{\NCth}{NC$^3$}
\newclass{\SACt}{SAC$^2$}
\newclass{\ShSACt}{\#SAC$^2$}
\newclass{\NC}{NC}
\newclass{\SAC}{SAC}
\newclass{\AC}{AC}
\newclass{\TC}{TC}
\newclass{\BPL}{BPL}
\newclass{\ZPP}{ZPP}
\newcommand{\RP}{\mathsf{RP}}
\newcommand{\TreeEval}{\mathsf{TreeEval}}
\newcommand{\pass}{\mathrm{pass}}

\title{An Operator Approach to Register Programs for Catalytic Computing}

\author{
    Antoine Vinciguerra\thanks{Supported by ISF grant 507/24.}  \\
    Technion Israel Institute of Technology \\
    \texttt{antoine.v@campus.technion.ac.il}
}
\date{}

\begin{document}
\maketitle
\thispagestyle{empty}

\begin{abstract}
In a seminal work, Buhrman et al.\ (STOC 2014) introduced catalytic computation and proved that uniform $\TCo$ circuits are computable in catalytic logspace, the class of problems solvable in space $s$ with an additional catalytic tape of size $c$, a tape whose initial content must be restored at the end of the computation.

A central ingredient of their proof is the register program model. Namely, they constructed a uniform family of register programs that computes $x^n$ using $n$ registers and four accesses to $x$. 

Since then, determining the number of registers and input accesses required to compute a polynomial of a given degree has become a central question in the study of catalytic computation. 

On one hand, we prove that the four-access bound of Buhrman et al.\ is optimal: every passive-output register program computing a polynomial of degree greater than three requires at least four input accesses, independently of the number of registers. 

On the other hand, we show that their register bound is not optimal. For every $t\geq2$ and every field $K$ of characteristic $0$ or greater than $2t-1$, we construct a register program for $x^{2t-1}$ with four input accesses and $t$ registers.

Our proofs rely on derivations and their exponential operators. This approach represents a register program as a series of exponential derivation operators, reducing register restoration to an operator identity.

Finally, we use the uniform family of register programs to improve known trade-offs for catalytic streaming algorithms and register programs for matrix powering.

The generalization of the lower-bound methods and the construction of the
uniform family of register programs were developed with assistance from
ChatGPT 5.6.
\end{abstract}

\clearpage
\pagenumbering{arabic}

\section{Introduction}
\label{sec:introduction}

\subsection{Catalytic computation}
In catalytic computing, one asks whether a space-bounded machine with limited
workspace can augment its computational power by using additional memory
containing arbitrary data, with the caveat that this memory must be completely
restored to its initial state at the end of the computation.

Buhrman, Cleve, Kouck\'y, Loff, and Speelman formalized this model through the
class $\CSPACE(s,c)$, in which a machine has $s$ bits of clean workspace and
$c$ bits of catalytic space~\cite{BuhrmanEtAl2014}. Catalytic Logspace is the
regime
\begin{equation*}
  \CL=\CSPACE\bigl(O(\!\log n),2^{O(\log n)}\bigr).
\end{equation*}
Prior to their work, it was unclear whether access to additional arbitrary
``garbage'' memory could increase the machine's computational power
\cite{cook2012pebbles,edmonds2018hardness,iwama2018read,liu2013pebbling}.
This seminal paper proves that $\TCo \subseteq \CL$, providing the first
evidence that catalytic memory is useful for computation. Indeed, this
containment implies that $\CL$ captures the determinant as well as the
nondeterministic class $\NL$, which is widely believed to lie outside
deterministic logspace $\L$.

Beyond general circuit classes, catalytic machines have also proved useful for
solving other classical problems in theoretical computer science. In fact, the
decision and search versions of Bipartite Matching, Linear Matroid
Intersection, and General Matching have all been placed in Catalytic Logspace
\cite{AgarwalaMertz2025,AgarwalaAlekseevVinciguerra2026,
arvind2025derandomizing,chakraborty2026maximum}. While recent breakthroughs
have placed the decision versions of Bipartite Matching and Linear Matroid
Intersection in $\TCo$~\cite{chatterjee2026bipartite,kopparty2026cggrt}, the
corresponding search variants are only known to be in $\TC^2$. Additionally,
the General Matching problem is currently only known to be in quasi-$\NC^3$.

Building on these results, the catalytic framework has been applied to other
settings, namely:
\begin{itemize}
  \item nondeterministic, unambiguous, randomized, and symmetric computation
    \cite{BuhrmanEtAl2018,GuptaEtAl2019,DattaEtAl2020,
    CookEtAl2025Structure,KouckyEtAl2025},
  \item relaxed catalytic regimes: lossy~\cite{GuptaEtAl2024Lossy,
    FolkertsmaEtAl2025}, robust~\cite{KouckyMertzSami2026}, and
    almost-catalytic computation~\cite{BisoyiEtAl2025},
  \item nonuniform and catalytic branching programs
    \cite{Potechin2017,RobereZuiddam2021,CookMertz2022}.
\end{itemize}
Additionally, catalytic models have been introduced for quantum computation
\cite{BuhrmanEtAl2025Quantum}, communication
\cite{PyneSheffieldWang2025}, streaming
\cite{BeckerEtAl2026,KaplanEtAl2026}, and function computation through FCL and
in-place FCL~\cite{CookEtAl2025InPlace}. See the surveys of Kouck\'y and Mertz
for more details~\cite{Koucky2016,Mertz2023}.

Catalytic methods have also led to progress in space-bounded derandomization.
Pyne proved
$\mathsf{BPSPACE}[S]\subseteq\mathsf{CSPACE}[S,S^2]$~\cite{Pyne2024}, while
Li, Pyne, and Tell showed unconditionally that derandomization in Catalytic
Logspace requires targeted pseudorandom generators~\cite{LiPyneTell2024}. In
parallel, a complementary propagation approach yielded a direct simulation of
$\BPL$~\cite{Cook2025BPL}. This propagation method was later extended to derive
efficient algorithms for directed connectivity, random walks, and various
string and distance metrics~\cite{CookPyne2026,ChmelEtAl2026}.

Finally, the catalytic paradigm has proved valuable outside the
catalytic-machine model itself. It has led to new branching programs and
space-efficient algorithms for $\TreeEval$
\cite{CookMertz2020,CookMertz2022,CookMertz2024}. These results, in turn,
are the main ingredient in Williams's recent square-root-space simulation of
time~\cite{Williams2025}.

This motivates the following question: which algebraic computations can or
cannot be implemented on a catalytic tape?

\subsection{Register programs}
The central tool behind Buhrman et al.'s result $\TCo \subseteq \CL$ is the
register program framework. This model traces back to the group programs of
Coppersmith and Grossman~\cite{CoppersmithGrossman1975} and was used to compare
formulas and arithmetic circuits with branching programs
\cite{Barrington1989,BenOrCleve1992}.

A register program over a ring $R$ with inputs $x_1,\dots,x_n\in R$ is defined
as a sequence of instructions $I_1,\dots,I_m\colon R^s\to R^s$ applied to a set
of registers $\{R_1,\dots,R_s\}$. The crucial property for catalytic computing
is that the program is reversible and ``clean'': upon termination, every work
register must be restored to its initial state, leaving the computed result in
a designated output register.

Buhrman et al. showed that clean register programs can be directly simulated
by a catalytic machine~\cite{BuhrmanEtAl2014}. Proving the inclusion
$\TCo\subseteq\CL$ thereby reduces to constructing efficient register programs
for $\TCo$ gates. They, in turn, reduced this task to constructing efficient
programs for powering and ultimately obtained a register program computing
$x^k$ using $k+2$ registers and only four accesses to $x$.

While the number of registers determines the required catalytic space, the
number of input accesses, or \emph{recursive calls}, controls the complexity of
recursive composition. Despite the development of register programs using a
constant number of recursive calls for multivariate polynomials and matrix
powering~\cite{AlekseevEtAl2025}, the exact limits of these constructions remain
unclear: given a fixed number of dirty registers and a fixed number of input
calls, how large a polynomial degree can a clean register program generate? In
other words, what is the optimal value of $D_r(t)$, where
\begin{equation*}
  D_r(t)
  =\sup\{\deg f : f\in K[x],
    f \text{ can be computed with at most } t \text{ calls and }
    r \text{ registers}\}.
\end{equation*}

\subsection{Our results}
\label{sec:our-results}

Our main contribution is a new compact operator notation for studying register
programs. Given registers $R_1,\dots,R_n$, we identify each register instruction
with its induced action on the polynomial ring in the register variables
$K[R_1,\dots,R_n]$. We then show how an input access can be interpreted as the
exponential $\exp(xD)$ of a derivation $D$. As a result, for a register program
$\mathcal R$ computing $f(x)$ with $t$ calls, we derive an equivalent
representation:
\begin{equation}
  \mathcal R \simeq \exp(xD_1)\cdots\exp(xD_t).
  \label{eq:operator-representation-intro}
\end{equation}

We use this viewpoint to derive our main results.

\medskip
We first prove that the register program of Buhrman et al. is essentially
optimal with respect to its number of input accesses.

\begin{theorem*}
Let $K$ be a field of characteristic $0$, and let $P\in K[x]$ with
$\deg(P)\ge 3$. If a register program $\mathcal R$ cleanly computes $P$, then
$\mathcal R$ must make at least $4$~input accesses.
\end{theorem*}

\medskip
We then prove that, within a natural restricted model defined later, the
interpolation register program of Cook and Mertz
\cite{CookMertz2024,goldreich2025cook} is also optimal with respect to the
number of input accesses.

\begin{theorem*}
Let $K$ be a field of characteristic $0$, and let $P\in K[x]$ with
$d=\deg(P)\ge 2$. If a passive-output register program $\mathcal R$ using two
registers cleanly computes $P$, then $\mathcal R$ must make at least
$d+1$~input accesses.
\end{theorem*}

\medskip
Additionally, we construct a logspace-uniform family of register programs
computing $x^{2r-1}$ with four input accesses and only $r$ registers, improving
upon the register bound achieved by Buhrman et al.

\begin{theorem*}
For every $r\ge 2$ and every field $K$ whose characteristic is either $0$ or a
prime $p>2r-1$, there exists a register program $\mathcal R$ that cleanly
computes $x^{2r-1}$ with $4$~input accesses and $r$~registers.
\end{theorem*}

The operator method is essential to this construction. We represent the
contents of the $r-1$ work registers as the non-leading coefficients of a monic
polynomial of degree $r-1$. We then choose derivations $A$, $B$, and $D$
satisfying
\begin{equation*}
  \exp(xA/2)\exp(xB)\exp(xA/2)\exp(-xD)
  =\operatorname{Id}.
\end{equation*}
The four corresponding input accesses therefore cancel on the work registers,
restoring each of them to its initial value. By inserting suitable polynomial
updates to the output register between these accesses, we make the program
compute $x^{2r-1}$.

\medskip
Finally, we give two direct applications of this new family of register
programs. We first obtain new pass--space trade-offs for computing frequency
moments in catalytic streaming, all of which improve upon previously known
space bounds. We also present register programs for matrix powering that reduce
both the number of recursive calls and the number of required registers.

\subsection{Organization}
Section~2 introduces clean and passive register programs and presents the
operator notation needed in the proofs. Section~3 proves the lower bounds and
determines the exact two-register frontier. Section~4 proves the uniform
four-call construction. Finally, Section~5 develops its applications to
catalytic computation, streaming algorithms, and matrix powering.

\section{Preliminaries}
\label{sec:preliminaries}

\subsection{Catalytic computation}
\label{sec:prelim-catalytic}

Catalytic computation was introduced by Buhrman et
al.~\cite{BuhrmanEtAl2014}. A catalytic Turing machine is a space-bounded
Turing machine with a read-only input tape and two read-write tapes: a work
tape of size $s=s(n)$ and a catalytic tape of size $c=c(n)$. The catalytic
tape initially contains an arbitrary string $\tau\in\{0,1\}^{c}$. The machine
may use and modify this string during its computation, but, for every input and
every $\tau$, it must halt with the catalytic tape restored exactly to $\tau$.

\begin{definition}[Catalytic space]
The class $\CSPACE(s,c)$ consists of the problems solvable by a catalytic
Turing machine with workspace $s(n)$ and catalytic space $c(n)$. Catalytic
logspace is
\begin{equation*}
  \CL=\CSPACE\!\left(O(\log n),2^{O(\log n)}\right).
\end{equation*}
\end{definition}

The bound $2^{O(s)}$ is a natural limit: addressing a single cell of a larger
catalytic tape would itself require more than $O(s)$ bits.

\subsection{Register programs}
\label{sec:prelim-register-programs}

Register programs were introduced by Coppersmith and
Grossman~\cite{CoppersmithGrossman1975} and later used by Ben-Or and
Cleve~\cite{BenOrCleve1992} to simulate arithmetic circuits. They serve as the
main algebraic framework in the catalytic simulation of Buhrman et
al.~\cite{BuhrmanEtAl2014}.

\begin{lemma}[{\cite[Lemma~15]{BuhrmanEtAl2014}}]
\label[lemma]{lem:register-program-catalytic-simulation}
Any clean register program of time $t$, using $s$ registers and $n$ inputs over
a finite ring $R$, can be simulated by a catalytic Turing machine with
\begin{equation*}
  O(\log t+\log n+\log |R|)
\end{equation*}
bits of clean space and
\begin{equation*}
  O(s\log |R|)
\end{equation*}
bits of catalytic space.
\end{lemma}

Throughout this paper, we follow the formulation of register programs used by
Alekseev et al.~\cite{AlekseevEtAl2025}.

\begin{definition}[Register program]
Let $R$ be a ring. An $R$-register program over inputs
$x_1,\ldots,x_n\in R$ consists of registers $R_1,\ldots,R_r$ and a sequence of
instructions of the following two forms:
\begin{align*}
  R_i&\leftarrow R_i+p(R_1,\ldots,R_{i-1},R_{i+1},\ldots,R_r),
      &&\text{basic instruction},\\
  R_i&\leftarrow R_i+\lambda x_j,\qquad \lambda\in R,
      &&\text{input access}.
\end{align*}
Here $p$ is an input-independent polynomial that does not involve $R_i$.
The \emph{time} is the total number of instructions, while the number of input
accesses is recorded separately. When the input of one program is the output
of another, an input access is called a \emph{recursive call}.
\end{definition}

A register program $P$ \emph{cleanly} computes $f\colon R^n\to R^m$ if there
exists a set of output registers $S\subseteq\{1,\dots,s\}$ with
$|S|=m$ such that, from any initial state $(R_1,\dots,R_s)$, the program
terminates with
\begin{itemize}
  \item $R_i\leftarrow R_i+f_i(x)$ for every $i\in S$;
  \item $R_j\leftarrow R_j$ for every $j\notin S$.
\end{itemize}

A register program is called \emph{passive} if the output registers never
occur on the right-hand side of an instruction. Thus, output registers collect
intermediate results and store the final output but cannot serve as additional
workspace.

For a field $K$, we write $f\in\RP_K(r,t)$ if $f$ has a clean program using
$r$ registers and at most $t$ input accesses, and
$f\in\RP_K^{\pass}(r,t)$ if the program is passive. For univariate polynomials,
we set
\begin{equation*}
  D^{\pass}_{r,K}(t)
  =\sup\{\deg f : f\in K[x]\text{ and }f\in\RP_K^{\pass}(r,t)\}.
\end{equation*}
We use the following passive-output version of the composition lemma from
Alekseev et al.~\cite{AlekseevEtAl2025}.

\begin{lemma}[{\cite[Lemma~8]{AlekseevEtAl2025}}]
\label{recycled-passive-composition}
Let $f\in\RP_K^{\pass}(r,t)$ and $g\in\RP_K^{\pass}(s,u)$. Then
\begin{equation*}
  f\circ g\in\RP_K^{\pass}\!\left(\max\{r,s+1\},tu\right).
\end{equation*}
Furthermore, if the programs for $f$ and $g$ contain $b_f$ and
$b_g$ basic instructions, respectively, then the constructed program for
$f\circ g$ uses $b_f+t b_g$ basic instructions.
\end{lemma}

The main construction of this paper improves the register bound achieved by
the four-input-access powering programs of Buhrman et
al.~\cite{BuhrmanEtAl2014}. We state it here for later use and prove it in
Section~4.

\begin{theorem}
\label{thm:uniform-four-call-power}
For every $r\ge 2$ and every field $K$ of characteristic $0$ or characteristic
$p>2r-1$,
\begin{equation*}
  x^{2r-1}\in\RP_K^{\pass}(r,4).
\end{equation*}
\end{theorem}

\subsection{Operator representations}
\label{sec:prelim-operators}

Throughout this subsection, $K$ has characteristic $0$. The standard identities
recalled below, including the Campbell--Baker--Hausdorff formula, can be found
in introductory texts (see, for example,~\cite[Ch.~2]{Hall2015}).

Let $M=K[\mathbf R,O]$ denote the polynomial ring in the machine's work-register
variables $\mathbf R=(R_1,\dots,R_{r-1})$ and output-register variable $O$.
Let $E=\operatorname{End}_K(M)$ be the algebra of $K$-linear endomorphisms of
$M$.

For $P,Q\in E$, write $[P,Q]=PQ-QP$. The notation $E[[x]]$ denotes the ring of
formal power series $\sum_{n\ge0}P_nx^n$ with coefficients $P_n\in E$.

For $P\in E$, define its formal exponential by
\begin{equation*}
  \exp(xP)=\sum_{n\ge0}\frac{x^n}{n!}P^n
  \in E[[x]].
\end{equation*}

\begin{lemma}
\label{exponentials_commutativity}
Let $P,Q\in E$ be linear maps. The following identities hold in $E[[x]]$:
\begin{align*}
  \exp(0)&=I,\\
  \exp(xP)^{-1}&=\exp(-xP),\\
  [P,Q]=0&\implies
    \exp(xP)\exp(xQ)=\exp\bigl(x(P+Q)\bigr).
\end{align*}
\end{lemma}
More generally, if $\Phi\colon M\to M$ is an invertible $K$-linear map, then
\begin{equation}
  (\Phi P\Phi^{-1})^n=\Phi P^n\Phi^{-1},
  \qquad
  \Phi\exp(xP)\Phi^{-1}
  =\exp\!\left(x\Phi P\Phi^{-1}\right).
  \label{eq:linear-conjugation-exponential}
\end{equation}

We finally state the Campbell--Baker--Hausdorff identity in the special case in
which $[P,Q]$ commutes with both $P$ and $Q$.

\begin{proposition}[Campbell--Baker--Hausdorff identity]
\label[proposition]{bch_commute}
If $[[P,Q],P]=[[P,Q],Q]=0$, then
\begin{equation*}
  \exp(xP)\exp(xQ)
  =\exp\left(x(P+Q)+\frac{x^2}{2}[P,Q]\right).
\end{equation*}
\end{proposition}

\subsection{Single-access operator representations}
\label{sec:prelim-one-call}
We introduce an
exponential-operator representation of input accesses that will later simplify
the algebraic analysis of both lower bounds and register-program constructions.

\begin{definition}[Derivations and local nilpotence]
Let $\delta\colon K[\mathbf R,O]\to K[\mathbf R,O]$. The map $\delta$ is a
$K$-derivation of the polynomial ring $K[\mathbf R,O]$ if it is $K$-linear and
satisfies the Leibniz rule
\begin{equation*}
  \delta(fg)=f\delta(g)+g\delta(f).
\end{equation*}
We say that $\delta$ is \emph{locally nilpotent} if, for every
$h\in K[\mathbf R,O]$, there is an integer $N=N(h)$ such that $\delta^N(h)=0$.
\end{definition}
For a detailed exposition of locally nilpotent derivations,
see~\cite{Freudenburg2017}.

The exponential of a locally nilpotent derivation is a finite sum when applied
to any fixed polynomial. We use the following consequence of the Leibniz rule.

\begin{lemma}
\label{exp_properties}
Let $\delta$ be a locally nilpotent derivation. For all $g,h\in M$,
\begin{equation*}
  \exp(x\delta)(gh)
  =\exp(x\delta)(g)\exp(x\delta)(h).
\end{equation*}
Moreover, $\exp(x\delta)$ is a polynomial automorphism with inverse
$\exp(-x\delta)$.
\end{lemma}

For a register variable $X$, let $\partial_X$ denote ordinary formal
differentiation with respect to $X$. Since $\partial_X(X)=1$,
$\partial_X^2(X)=0$, and $\partial_X(Y)=0$ for $Y\ne X$, the exponential series
gives
\begin{equation}
  \exp(\lambda x\partial_X)(X)=X+\lambda x,
  \qquad
  \exp(\lambda x\partial_X)(Y)=Y\quad(Y\ne X).
  \label{partial_deriv_translation}
\end{equation}

Thus, a translation of register $X$ by $\lambda x$ is represented by the
operator $\exp(\lambda x\partial_X)$. This observation extends to a larger class
of derivations.

\begin{lemma}
\label[lemma]{lem:one-call-implementation}
Let $\delta$ be a locally nilpotent derivation acting on the registers
$(O,Y_1,\ldots,Y_m)$. Suppose that $\delta(O)=c\in K^*$,
$\delta(Y_1)\in K[O]$, and, for every $2\le i\le m$,
\begin{equation*}
  \delta(Y_i)\in K[O,Y_1,\ldots,Y_{i-1}].
\end{equation*}
Then $\exp(x\delta)$ is equivalent, up to a triangular polynomial automorphism,
to the translation $O\leftarrow O+cx$ and can be implemented with one input
access.
\end{lemma}

\begin{proof}
Let $E$ denote the polynomial transformation
\begin{equation*}
  E\colon (O,Y)\longmapsto\bigl(\exp(x\delta)(O),\exp(x\delta)(Y)\bigr)
\end{equation*}
of the registers. For every $i\ge1$, define
\begin{equation*}
  V_i=\sum_{n\ge0}\frac{(-O/c)^n}{n!}\delta^n(Y_i),
\end{equation*}
where the sum is finite by local nilpotence.

Using $\delta(O)=c$ and the Leibniz rule, we obtain
\begin{equation*}
  \delta(V_i)
  =-\sum_{n\ge1}\frac{(-O/c)^{n-1}}{(n-1)!}\delta^n(Y_i)
   +\sum_{n\ge0}\frac{(-O/c)^n}{n!}\delta^{n+1}(Y_i)
  =0.
\end{equation*}
Since each $\delta(Y_i)$ depends only on $(O,Y_1,\ldots,Y_{i-1})$, we can write
$V_i=Y_i+h_i(O,Y_1,\ldots,Y_{i-1})$. Therefore, the polynomial transformation
\begin{equation*}
  \Pi\colon(O,Y_1,\ldots,Y_m)\longmapsto(O,V_1,\ldots,V_m)
\end{equation*}
is a triangular polynomial automorphism.

Because $\exp(x\delta)$ is a ring homomorphism, evaluating $\Pi$ on the updated
registers gives
\begin{align*}
  (\Pi\circ E)(O,Y)
  &=\Pi\bigl(\exp(x\delta)(O),\exp(x\delta)(Y)\bigr)\\
  &=\bigl(\exp(x\delta)(O),\exp(x\delta)(V_1),\ldots,
      \exp(x\delta)(V_m)\bigr).
\end{align*}
Moreover, since $\delta(V_i)=0$, applying $\exp(x\delta)$ yields
\begin{align*}
  \exp(x\delta)(O)&=O+cx,\\
  \exp(x\delta)(V_i)&=V_i.
\end{align*}
Hence
\begin{equation*}
  (\Pi\circ E)(O,Y)=(O+cx,V_1,\ldots,V_m).
\end{equation*}
The right-hand side is obtained by applying $\Pi$ and then translating $O$ by
$cx$. Thus, $\Pi\circ E\circ\Pi^{-1}$ is a translation.
\end{proof}

In the passive-output model, a clean program computing a polynomial $f(x)$ acts
precisely as the exponential operator $\exp(f(x)\partial_O)$, which is the
translation $O\leftarrow O+f(x)$. The next section uses this observation to put
any clean passive program into a convenient operator form.

\section{Call Lower Bounds and Two-Register Programs}
\label{sec:lower-bounds}

Throughout this section, $K$ is a field of characteristic $0$.

We first characterize what can be computed with up to three calls when there is
no restriction on the number of registers. We then focus on the two-register
setting, where we establish the exact degree--call trade-off. Throughout this
section, we rely on the formal exponential identities from
\Cref{sec:prelim-operators} and the single-access transformations from
\Cref{sec:prelim-one-call}.

\subsection{Product representation of passive programs}
\label{sec:program-product}

Recall that the operator $\exp(f(x)\partial_O)$ represents the translation
$O\mapsto O+f(x)$ of the output register. We use this operator notation to rewrite a register program as a product of exponentials. 
\begin{lemma}
\label[lemma]{lem:program-product}
Let $f\in K[x]$, and let $\mathcal R$ be a passive-output register program that
cleanly computes $f$ with $t$ input accesses over registers $\mathbf R$.
There exist locally nilpotent derivations $D_1,\ldots,D_t$ of
$M=K[\mathbf R,O]$ such that
\begin{equation}
  \mathcal R=\exp(xD_1)\cdots\exp(xD_t)
  \qquad\text{and}\qquad
  [D_i,\partial_O]=0 \quad (1\le i\le t).
  \label{eq:central-output-normal-form}
\end{equation}
\end{lemma}

\begin{proof}
We may assume without loss of generality that $f(0)=0$, since we can first apply
the call-free instruction $O\leftarrow O-f(0)$.

By grouping consecutive call-free instructions, we can write the program
$\mathcal R(x)$ as
\begin{equation*}
  \mathcal R(x)=B_0T_1(x)B_1\cdots T_t(x)B_t,
\end{equation*}
where each $B_i$ is a polynomial automorphism and
$T_i(x)=\exp(x\lambda_i\partial_{R_{j_i}})$ represents an input access to
register $R_{j_i}$.

Since the program is clean and $f(0)=0$, setting $x=0$ yields
$\mathcal R(0)=I$, and hence $B_0B_1\cdots B_t=I$. For $1\le i\le t$, define $C_i=B_0\cdots B_{i-1}$. Using
$B_0\cdots B_t=I$, we can rewrite $\mathcal R(x)$:
\begin{align*}
  \mathcal R(x)
  ={}&(B_0T_1(x)B_0^{-1})
      (B_0B_1T_2(x)B_1^{-1}B_0^{-1})\\
    &\quad\cdots
      (B_0\cdots B_{t-1}T_t(x)B_{t-1}^{-1}\cdots B_0^{-1})
      B_0\cdots B_t\\
  ={}&\prod_{i=1}^t C_iT_i(x)C_i^{-1}.
\end{align*}
Applying the exponential conjugation identity
$C_i\exp(X)C_i^{-1}=\exp(C_iXC_i^{-1})$, we obtain
\begin{equation*}
  \mathcal R(x)
  =\prod_{i=1}^t
    \exp\!\left(x\lambda_iC_i\partial_{R_{j_i}}C_i^{-1}\right)
  =\prod_{i=1}^t\exp(xD_i),
\end{equation*}
where $D_i=\lambda_iC_i\partial_{R_{j_i}}C_i^{-1}$.

It remains to verify the two required properties of $D_i$:
\begin{enumerate}
  \item \textbf{Local nilpotence.}
    The formal partial derivative $\partial_{R_{j_i}}$ is locally nilpotent. Since
    $D_i^n=\lambda_i^nC_i\partial_{R_{j_i}}^nC_i^{-1}$ for every $n\ge1$,
    $D_i$ is also locally nilpotent.

  \item \textbf{Commutation with the output derivative.}
    Since the program is passive, the work registers are updated independently
    of $O$. Thus, both $C_i$ and $C_i^{-1}$ commute with $\partial_O$.
    Moreover, the distinct partial derivatives $\partial_{R_{j_i}}$ and
    $\partial_O$ commute. It follows that
    $D_i\partial_O=\partial_OD_i$, and hence $[D_i,\partial_O]=0$.
\end{enumerate}
\end{proof}

\subsection{Lower bounds for one, two, and three calls}
\label{sec:universal-call-bounds}

We are now ready to prove our first lower bounds. We begin by showing that,
regardless of the number of registers, passive register programs with at most
two input accesses can compute only polynomials of degree at most one.

\begin{theorem}
\label{thm:one-two-calls}
For every $r\ge1$,
\begin{equation*}
  D^{\pass}_{r,K}(1)=D^{\pass}_{r,K}(2)=1.
\end{equation*}
\end{theorem}

\begin{proof}
Let $f$ be cleanly computed by a passive register program with at most two
input accesses. We assume without loss of generality that $f(0)=0$.
Applying $\partial_x$ to the exponential operators from
\Cref{lem:program-product} gives
\begin{equation}
  \partial_x\exp(xD)
  =\partial_x\sum_{n\ge0}\frac{x^n}{n!}D^n
  =\sum_{n\ge1}\frac{x^{n-1}}{(n-1)!}D^n
  =D\exp(xD)
  \label{eq:partial_x_D}
\end{equation}
and
\begin{equation}
  \partial_x\exp(f(x)\partial_O)
  =f'(x)\partial_O\exp(f(x)\partial_O).
  \label{eq:partial_x_f_x}
\end{equation}
We treat the one-call and two-call cases separately.

\begin{itemize}
  \item \textbf{One call.}
    Let $c=f'(0)$. By \Cref{lem:program-product},
    \begin{equation*}
      \exp(xD_1)=\exp(f(x)\partial_O).
    \end{equation*}
    Applying $\partial_x$ and evaluating at $x=0$ using
    \Cref{eq:partial_x_D,eq:partial_x_f_x} yields $D_1=c\partial_O$.
    Therefore,
    \begin{equation*}
      \exp(cx\partial_O)=\exp(f(x)\partial_O).
    \end{equation*}
    Evaluating both operators on the output register $O$ gives
    $O+cx=O+f(x)$, so $f(x)=cx$.

  \item \textbf{Two calls.}
    The same argument applies. By \Cref{lem:program-product},
    \begin{equation*}
      \exp(xD_1)\exp(xD_2)=\exp(f(x)\partial_O),
    \end{equation*}
    and applying $\partial_x$ at $x=0$ gives
    $D_1+D_2=c\partial_O$. We can rewrite this identity as
    $D_2=-D_1+c\partial_O$. Since $[D_1,\partial_O]=0$ by
    \Cref{lem:program-product}, the derivations $D_1$ and $D_2$ commute.
    Using \Cref{exponentials_commutativity}, we obtain
    \begin{align*}
      \exp(f(x)\partial_O)
      &=\exp(xD_1)\exp(xD_2)\\
      &=\exp\bigl(x(D_1+D_2)\bigr)\\
      &=\exp(cx\partial_O).
    \end{align*}
    Once again, evaluating both operators on $O$ gives
    $O+cx=O+f(x)$, so $f(x)=cx$.
\end{itemize}
Since every nonzero linear polynomial can be computed with a single input
access to $O$, we conclude that
$D^{\pass}_{r,K}(1)=D^{\pass}_{r,K}(2)=1$.
\end{proof}

We now show that a third input access offers limited additional power: it
allows the computation of polynomials of degree at most two.

\begin{theorem}
\label{thm:three-calls}
For every $r\ge1$, $D^{\pass}_{r,K}(3)\le2$. Equality holds for every $r\ge2$.
\end{theorem}

\begin{proof}
Let $f$ be cleanly computed by a passive register program with at most three
input accesses. We may assume without loss of generality that $f(0)=0$.

By \Cref{lem:program-product}, there exist locally nilpotent derivations
$D_1,D_2,D_3$ such that
\begin{equation*}
  \exp(xD_1)\exp(xD_2)\exp(xD_3)=\exp(f(x)\partial_O),
  \qquad
  [D_i,\partial_O]=0.
\end{equation*}
Set $c=f'(0)$. Applying $\partial_x$ to this equation and evaluating at $x=0$
gives
\begin{equation*}
  D_1+D_2+D_3=c\partial_O,
\end{equation*}
which we rewrite as $D_3=-(D_1+D_2)+c\partial_O$. Since $\partial_O$ commutes
with $D_1+D_2$, the commuting-exponential identity yields
\begin{equation*}
  \exp(xD_3)
  =\exp\bigl(-x(D_1+D_2)\bigr)\exp(cx\partial_O).
\end{equation*}
Substituting this expression into the product identity and canceling
$\exp(cx\partial_O)$ gives
\begin{equation}
  \exp(xD_1)\exp(xD_2)\exp\bigl(-x(D_1+D_2)\bigr)
  =\exp\bigl((f(x)-cx)\partial_O\bigr).
  \label{eq:three-call-commutator}
\end{equation}

We now compare the coefficients of $x^2$ on both sides. Expanding the
exponential series on the left-hand side gives
\begin{align*}
  &[x^2]\exp(xD_1)\exp(xD_2)\exp\bigl(-x(D_1+D_2)\bigr)\\
  &\qquad={}
    \frac{D_1^2}{2}+D_1D_2+\frac{D_2^2}{2}
    -(D_1+D_2)^2+\frac{(D_1+D_2)^2}{2}\\
  &\qquad={}
    \frac12\bigl(D_1D_2-D_2D_1\bigr)
  =\frac12[D_1,D_2].
\end{align*}
On the other hand, since
$f(x)-cx=\frac{f''(0)}{2}x^2+O(x^3)$, expanding the right-hand side gives
\begin{equation*}
  [x^2]\exp\bigl((f(x)-cx)\partial_O\bigr)
  =\frac{f''(0)}{2}\partial_O.
\end{equation*}
Therefore, writing $e=f''(0)$, we obtain the commutator relation
\begin{equation*}
  [D_1,D_2]=e\partial_O.
\end{equation*}

Since $[D_1,\partial_O]=[D_2,\partial_O]=0$, the commutator
$[D_1,D_2]=e\partial_O$ commutes with both $D_1$ and $D_2$. Applying
\Cref{bch_commute} yields
\begin{equation*}
  \exp(xD_1)\exp(xD_2)
  =\exp\left(x(D_1+D_2)+\frac{e x^2}{2}\partial_O\right).
\end{equation*}
Finally, since $\partial_O$ commutes with $D_1+D_2$, substituting this identity
into \Cref{eq:three-call-commutator} and using
\Cref{exponentials_commutativity} gives
\begin{equation*}
  \exp\bigl((f(x)-cx)\partial_O\bigr)
  =\exp\left(\frac{e x^2}{2}\partial_O\right).
\end{equation*}
Applying both sides to the output register $O$ yields
\begin{equation*}
  f(x)=cx+\frac{ex^2}{2},
\end{equation*}
and hence $\deg f\le2$.

Finally, the interpolation register program of Cook and Mertz computes $x^2$
with two registers and three calls~\cite{CookMertz2024}. Hence, the bound is
attained for every $r\ge2$.
\end{proof}

\begin{corollary}
\label{cor:four-calls-necessary}
Every passive-output register program that cleanly computes a polynomial of
degree at least $3$ requires at least $4$ input accesses, independently of the
number of registers.
\end{corollary}

\noindent
This implies that, for every degree $d>2$, the register-program construction of
Buhrman et al.~\cite{BuhrmanEtAl2014} is optimal in terms of input accesses.

\subsection{The two-register case}
\label{sec:width-two-frontier}

In the previous subsection, we proved that four input accesses are necessary
for a passive register program to compute a polynomial of degree $3$ over a
field $K$ of characteristic $0$. The next natural question is what changes when
we restrict the number of available registers.

We consider passive programs with only two registers. In this setting, there
is one work register $U$ and one passive output register $O$. Immediately after
the $i$th input access, the content of $U$ has the form
$\tau+c+ax$, where $\tau$ is the initial value of $U$ and $ax$ is the
accumulated contribution of the input accesses. We first show that the number
of distinct values of $a$ is restricted.

\begin{lemma}
\label[lemma]{lem:vandermonde-coefficient-bound}
Let $A\subset K$ be finite, let $H_a\in K[z]$ for every $a\in A$, and suppose
that
\begin{equation*}
  \sum_{a\in A}H_a(\tau+ax)=F(x),
\end{equation*}
where $F$ has degree $d\ge2$. Then $|A|\ge d+1$.
\end{lemma}

\begin{proof}
Set $m=|A|$ and $M=\max_{a\in A}\deg H_a$. Necessarily, $M\ge d$.

For every $a\in A$, let $h_a$ denote the coefficient of $z^M$ in $H_a$.
Expanding the homogeneous component of total degree $M$ in
$\sum_{a\in A}H_a(\tau+ax)$ gives
\begin{equation*}
  \sum_{a\in A}H_a(\tau+ax)
  =\sum_{i=0}^M\binom{M}{i}\tau^{M-i}x^i
    \left(\sum_{a\in A}h_aa^i\right)
   +G(x,\tau),
\end{equation*}
where $G$ has total degree less than $M$.

Since $F(x)$ is independent of $\tau$, the coefficient of every monomial
$\tau^{M-i}x^i$ with $M-i\ge1$ must vanish. Hence,
\begin{equation*}
  \sum_{a\in A}h_aa^i=0 \qquad (0\le i<M).
\end{equation*}

If $m\le M$, the first $m$ equations, corresponding to $0\le i<m$, form a
nonsingular $m\times m$ Vandermonde system on the distinct elements of $A$.
They therefore force $h_a=0$ for every $a\in A$, contradicting the definition
of $M$.

Thus, $m\ge M+1$. Since $M\ge d$, we conclude that $m\ge d+1$.
\end{proof}

We can now prove the main theorem for two registers: computing a polynomial of
degree $d\ge2$ requires at least $d+1$ calls.

\begin{theorem}
\label{thm:width-two-frontier}
For every $t\ge2$,
\begin{equation*}
  D^{\pass}_{2,K}(t)=t-1.
\end{equation*}
\end{theorem}

\begin{proof}
Let the work register $U$ be initialized to $\tau$. Since the output is passive,
a call-free instruction acting on $U$ can only add a constant. Thus, whenever
the program modifies the output, the work register has the form
\begin{equation*}
  U=\tau+c+ax.
\end{equation*}
We can write every basic instruction that modifies the output as
\begin{equation*}
  O\longmapsto O+P(\tau+c+ax)
  =O+\widetilde P(\tau+ax),
\end{equation*}
where the coefficients of $\widetilde P$ depend on $c$. Input accesses to $O$
contribute only a linear polynomial. After grouping the basic instructions by
the coefficient $a$, we therefore obtain
\begin{equation*}
  \sum_{a\in A}H_a(\tau+ax)=f(x)-\ell(x),
\end{equation*}
where $A$ is the set of distinct possible values of $a$ and $\ell$ is linear.

Assume that $d=\deg f\ge2$. The right-hand side still has degree $d$, so
\Cref{lem:vandermonde-coefficient-bound} gives $|A|\ge d+1$. On the other hand,
each value in $A$ must be reached by an input access, so the number $t$ of input
accesses satisfies $t\ge|A|$. Thus, $t\ge d+1$, which proves that $d\le t-1$.

The matching construction is the Lagrange-interpolation register program of
Cook and Mertz~\cite{CookMertz2024,goldreich2025cook}.
\end{proof}

\section{A Four-Call Register Program for $x^{2r-1}$}
\label{sec::uniform-four-call-program}

The main result of this section is
\begin{equation}
  x^{2r-1}\in\RP_K^{\pass}(r,4)
  \qquad(r\ge2),
  \label{eq:uniform-four-call-result}
\end{equation}
over every characteristic-zero field $K$. Thus
$D^{\pass}_{r,K}(4)\ge2r-1$.

The construction encodes the contents of the $r-1$ work registers as the
coefficients of a monic polynomial $P(u)$ of degree $r-1$. We then define the
derivations underlying the four input accesses by their action on $P(u)$.
These derivations satisfy the conditions of \Cref{lem:one-call-implementation},
so each can be implemented using one input access. We prove that the four
resulting transformations form a closed loop and thus restore all the work
registers to their initial values. Finally, we insert suitable polynomial
updates to the output register around these transformations so that the
program computes $x^{2r-1}$.

\subsection{The Four Transformations and the Closed-Loop Construction}
\label{sec:four-derivations-loop}
Let $O$ be the output register, and let $R_1,\dots,R_{r-1}$ be the work
registers. Fix $r\ge3$ and suppose that $R_i$ initially contains $y_i\in K$.
Define
\begin{equation*}
  P(u)=u^{r-1}+y_1u^{r-2}+\cdots+y_{r-1}.
\end{equation*}

We define the derivations $B$ and $L$ on $K[y_1,\dots,y_{r-1}]$ by their
actions on $P(u)$:
\begin{equation}
  B(P) = -P'(u),
  \qquad
  L(P) = (r-1) u^{r-2}.
  \label{eq:BL-definitions}
\end{equation}
We set $A = L - 2B$ and $D = L - B$.

\begin{proposition}
\label{prop:triangularity}
The derivations $B$, $A/2$, and $-D$ are triangular and locally nilpotent.
Furthermore, their evaluations on $y_1$ are nonzero constants:
\begin{equation*}
  B(y_1) = -(r-1),
  \qquad
  \frac{A}{2}(y_1) = \frac{3(r-1)}{2},
  \qquad
  (-D)(y_1) = -2(r-1).
\end{equation*}
\end{proposition}

\begin{proof}
The definitions in \eqref{eq:BL-definitions} give
\begin{align*}
  B(y_j) &= -(r-j) y_{j-1}, \\
  L(y_1) &= r-1, \\
  L(y_j) &= 0 \quad (2 \le j \le r-1).
\end{align*}
where $y_0=1$. In particular, both $B$ and $L$ map each $y_j$ into
$K[y_1,\dots,y_{j-1}]$. The same is therefore true of their linear
combinations $A/2$ and $-D$. Hence, all three derivations are triangular and
locally nilpotent.

Evaluating these derivations on $y_1$ yields
\begin{align*}
  B(y_1) &= -(r-1) y_0 = -(r-1),
  \\
  \frac{A}{2}(y_1) &= \frac{L(y_1)}{2} - B(y_1) = \frac{r-1}{2} - (-(r-1)) = \frac{3(r-1)}{2},
  \\
  (-D)(y_1) &= -(L(y_1) - B(y_1)) = -((r-1) - (-(r-1))) = -2(r-1).
\end{align*}
These constants are nonzero since $r \ge 3$.
\end{proof}

Each derivation is triangular and locally nilpotent and evaluates to a nonzero
constant on $y_1$. Therefore, all three satisfy the hypotheses of
\Cref{lem:one-call-implementation}, so each corresponding transformation uses
exactly one input access.

\begin{lemma}
\label[lemma]{lem:strang-identity}
For every $x$,
\begin{align*}
  \exp(xB)P(u)
    &= P(u-x), \\
  \exp(xA/2)P(u)
    &= P(u+x) + \frac{(u+x)^{r-1} - u^{r-1}}{2}, \\
  \exp(xD)P(u)
    &= P(u+x) + (u+x)^{r-1} - u^{r-1}.
\end{align*}
Hence,
\begin{equation*}
  \exp(xA/2)\exp(xB)\exp(xA/2) = \exp(xD).
\end{equation*}
\end{lemma}

\begin{proof}
Let $\partial_u$ denote formal differentiation with respect to $u$. For any
polynomial $f(u)$ and scalar $h$, Taylor's formula gives
\begin{equation*}
  \sum_{n\ge 0} \frac{h^n}{n!} \partial_u^n f(u) = f(u+h).
\end{equation*}
For $B$, we have $B(P) = -\partial_u P$, which implies
$B^n(P) = (-\partial_u)^n P$ for all $n \ge 0$. The Taylor identity yields
\begin{equation*}
  \exp(xB)P(u) = \sum_{n\ge 0}\frac{x^n}{n!}(-\partial_u)^n P(u) = P(u-x).
\end{equation*}

For $A/2$, since
$(A/2)(P) = \partial_u P + \frac{1}{2}\partial_u(u^{r-1})$ and $u^{r-1}$
contains none of the register variables $y_j$, repeated application gives
$(A/2)^n(P) = \partial_u^n P + \frac{1}{2}\partial_u^n(u^{r-1})$ for all
$n \ge 1$. Summing the series yields
\begin{align*}
  \exp(xA/2)P(u)
    &= P(u) + \sum_{n\ge 1}\frac{x^n}{n!}\partial_u^n P(u) \\
    &\quad + \frac{1}{2}\sum_{n\ge 1}\frac{x^n}{n!}
      \partial_u^n(u^{r-1}) \\
    &= P(u+x) + \frac{(u+x)^{r-1} - u^{r-1}}{2}.
\end{align*}

Similarly, $D(P) = \partial_u P + \partial_u(u^{r-1})$ gives
\begin{equation*}
  \exp(xD)P(u) = P(u+x) + (u+x)^{r-1} - u^{r-1}.
\end{equation*}
Finally, composing these exponential operators gives
\begin{align*}
  P_1(u) &= \exp(xA/2)P(u) = P(u+x) + \frac{(u+x)^{r-1} - u^{r-1}}{2}, \\
  P_2(u) &= \exp(xB)P_1(u) = P(u) + \frac{u^{r-1} - (u-x)^{r-1}}{2}, \\
  P_3(u) &= \exp(xA/2)P_2(u)
           = P(u+x) + \frac{(u+x)^{r-1} - u^{r-1}}{2}
             + \frac{(u+x)^{r-1} - u^{r-1}}{2} \\
         &= P(u+x) + (u+x)^{r-1} - u^{r-1},
\end{align*}
which is exactly $\exp(xD)P(u)$.
\end{proof}

The four exponential transformations
\begin{equation*}
  \exp(xA/2),\quad
  \exp(xB),\quad
  \exp(xA/2),\quad
  \exp(-xD)
\end{equation*}
thus form a closed loop.

\subsection{Polynomial Updates for the Output Register}
\label{sec:output-updates}

Let $P_0(u) = P(u)$ be the initial polynomial state, and let $P_1,P_2,P_3$
denote the successive states produced by the four-step program. Each $P_i$ is
a monic polynomial of degree $r-1$, which we write as
\begin{equation*}
  P_i(u)=u^{r-1}+y_1^{(i)}u^{r-2}+\cdots+y_{r-1}^{(i)}.
\end{equation*}

Substituting $u = v + h_i$ gives
$P_i(v+h_i) = v^{r-1} + \bigl((r-1)h_i + y_1^{(i)}\bigr)v^{r-2}
+ O(v^{r-3})$. We define the center
$h_i = -\frac{y_1^{(i)}}{r-1}$, which makes the coefficient of the
degree-$(r-2)$ term vanish. The corresponding centered polynomial is
\begin{equation}
  C_i(v) = P_i(v + h_i)
  = v^{r-1} + c_2^{(i)}v^{r-3} + \cdots + c_{r-1}^{(i)}.
  \label{eq:centered-polynomial}
\end{equation}

The change of variables
$(y_1^{(i)}, \dots, y_{r-1}^{(i)}) \leftrightarrow
(h_i, c_2^{(i)}, \dots, c_{r-1}^{(i)})$ is triangular and invertible over
any field in which $r-1$ is invertible. In what follows, let $h = h_0$ and
$\mathbf{c} = \mathbf{c}^{(0)}$. We assign the following weights:
\begin{equation*}
  \operatorname{w}(h) = 1,
  \qquad
  \operatorname{w}(c_j) = j,
  \qquad
  \operatorname{w}(x) = 1.
\end{equation*}

Set $N = 2r - 1$, and let $a = r-1$. The update polynomials are chosen from
the vector space $\mathcal{V}_{a}$ of homogeneous polynomials of weight $N$
and degree at most two in the centered coefficients:
\begin{align}
  \mathcal{V}_{a} ={} & \operatorname{span}_K\{h^N\}
  \oplus \operatorname{span}_K \{c_i h^{N-i} : 2 \le i \le a\} \nonumber \\
  & \oplus \operatorname{span}_K \{c_i c_j h^{N-i-j} : 2 \le i \le j \le a\},
  \label{eq:potential-space}
\end{align}
which has dimension $a(a+1)/2$.

To compute the final output, we insert updates determined by two polynomials
$\Phi, \Psi \in \mathcal{V}_{a}$ between the four steps of the program:
\begin{align}
  \mathcal{T}(\Phi, \Psi)
  ={}& \Phi(P_1) - \Phi(P_0) \nonumber \\
  & + \Psi(P_2) - \Psi(P_1) \nonumber \\
  & + \Phi(P_3) - \Phi(P_2).
  \label{eq:output-increment-map}
\end{align}

Our goal is to choose these polynomials so that the total increment equals
$x^{2r-1}$. This amounts to solving a linear system in $h$ and the
coefficients $c_j$.

Since $P_2(u) = P_1(u-x)$, the coefficient of $u^{r-2}$ in $P_2$ is
$y_1^{(2)}=y_1^{(1)}-(r-1)x$. Recalling that
$h_i = -\frac{y_1^{(i)}}{r-1}$, we obtain $h_2 = h_1 + x$. Hence,
\begin{equation*}
  C_2(v) = P_2(v + h_2) = P_1(v + h_2 - x)  = P_1(v + h_1) = C_1(v),
\end{equation*}
Thus, $\mathbf{c}^{(1)} = \mathbf{c}^{(2)} = \widehat{\mathbf{c}}$. More
generally, the centers and centered coefficients satisfy
\begin{align*}
  (h_0, \mathbf{c}^{(0)}) &= (h, \mathbf{c}), \\
  (h_1, \mathbf{c}^{(1)}) &= (h - 3x/2, \widehat{\mathbf{c}}), \\
  (h_2, \mathbf{c}^{(2)}) &= (h - x/2, \widehat{\mathbf{c}}), \\
  (h_3, \mathbf{c}^{(3)}) &= (h - 2x, \widetilde{\mathbf{c}}),
\end{align*}
where the coefficients are given by
\begin{align}
  \widehat{c}_j
    ={}& \sum_{\ell=2}^{j} \binom{r-1-\ell}{j-\ell} \left(-\frac{x}{2}\right)^{j-\ell} c_\ell \nonumber \\
  & + \binom{r-1}{j} \left[ \left(-\frac{x}{2}\right)^j + \frac{(h-x/2)^j - (h-3x/2)^j}{2} \right],
  \label{eq:chat-coefficients} \\
  \widetilde{c}_j
    ={}& \sum_{\ell=2}^{j} \binom{r-1-\ell}{j-\ell} (-x)^{j-\ell} c_\ell \nonumber \\
  & + \binom{r-1}{j} \left[ (-x)^j + (h-x)^j - (h-2x)^j \right].
  \label{eq:ctilde-coefficients}
\end{align}

Rewriting $\mathcal{T}(\Phi, \Psi)$ in these centered coordinates yields
\begin{align}
  \mathcal{T}(\Phi, \Psi)
  ={}& \Phi(h - 3x/2, \widehat{\mathbf{c}}) - \Phi(h, \mathbf{c}) \nonumber \\
  & + \Psi(h - x/2, \widehat{\mathbf{c}}) - \Psi(h - 3x/2, \widehat{\mathbf{c}}) \nonumber \\
  & + \Phi(h - 2x, \widetilde{\mathbf{c}}) - \Phi(h - x/2, \widehat{\mathbf{c}}).
  \label{eq:output-increment-expanded}
\end{align}

\begin{lemma}
\label[lemma]{lem:transgression}
For every $r \ge 3$, there exists a unique pair
$(\Phi_r, \Psi_r) \in \mathcal{V}_{r-1}^2$ such that
\begin{equation*}
  \mathcal{T}(\Phi_r, \Psi_r) = x^{2r-1}.
\end{equation*}
The coefficients of $\Phi_r$ and $\Psi_r$ are rational and can be constructed
uniformly from binomial coefficients.
\end{lemma}

The proof of \Cref{lem:transgression} is given in Appendix B.

\subsection{The Final Four-Call Construction}
\label{sec:four-call-theorem}

We can now complete the construction of the register program.

\begin{proof}[Proof of \Cref{thm:uniform-four-call-power}]
For $r=2$, the result follows directly from \Cref{thm:width-two-frontier}.

For $r \ge 3$, the register program executes the following sequence of
instructions:
\begin{align*}
  O &\gets O - \Phi_r(P_0), \\
  &\quad \exp(xA/2), \\
  O &\gets O - \Psi_r(P_1) + \Phi_r(P_1), \\
  &\quad \exp(xB), \\
  O &\gets O - \Phi_r(P_2) + \Psi_r(P_2), \\
  &\quad \exp(xA/2), \\
  O &\gets O + \Phi_r(P_3), \\
  &\quad \exp(-xD).
\end{align*}
The correctness and cleanliness of this program follow directly from
\Cref{lem:one-call-implementation}, \Cref{lem:strang-identity}, and
\Cref{lem:transgression}.

Finally, as explained in Appendix A, the explicit formulas for $\Phi_r$ and $\Psi_r$ show that all their
coefficients belong to $\mathbb{Z}\left[\frac{1}{(2r-1)!}\right]$. In
particular, the reciprocal of the normalization factor
$\gamma = (-1)^{r-1}(2r-1)\binom{2r-2}{r-1}$ belongs to this localization.
Thus, the program is valid over any field $K$ with $\mathrm{char}(K) = 0$ or
$\mathrm{char}(K) > 2r-1$.
\end{proof}

The power of register programs lies in the efficiency with which they can be
composed. For example, composing the program for $x^{2r-1}$ with itself
computes $(x^{2r-1})^{2r-1} = x^{(2r-1)^2}$ using $16$ input accesses. This
idea is developed further in~\cite{AlekseevEtAl2025}. We adapt the powering
lemma from~\cite{AlekseevEtAl2025}, with a simple refinement, to compute $x^N$
for arbitrary $N$.

\begin{corollary}
\label[corollary]{cor:four-call-sequential-powering}
Let $r \ge 2$, let $\delta = 2r-1$, and let
$\epsilon(r) = \log_{\delta}4$. For every integer $N \ge 1$,
\begin{equation*}
  x^N \in \RP_K^{\pass}\left(r+1, \, 4\delta N^{\epsilon(r)}\right).
\end{equation*}
\end{corollary}

\begin{proof}
Let $R_1,\dots,R_{r-1},U,O$ be the registers. We use clean passive
subprograms $C_i$ on $R_1,\dots,R_{r-1},U$, with $U$ as their output
register. The register $O$ is the passive output register of the main program.

Expand $N$ in base $\delta$ as
\begin{equation*}
  N = \sum_{i=0}^{\ell} b_i \delta^i, \qquad 0 \le b_i < \delta.
\end{equation*}
By \Cref{thm:width-two-frontier}, each term $x^{b_i}$ can be evaluated using at
most $\delta$ input accesses. By composing this program with $i$ recursive
applications of the degree-$\delta$ four-call block via
\Cref{recycled-passive-composition}, we obtain a clean program $C_i$ that
computes $x^{b_i\delta^i}$ in $U$, uses $r$ registers, and makes at most
$q_i \le \delta \cdot 4^i$ input accesses.

To compute the product $\prod_{i=0}^{\ell} x^{b_i\delta^i}$, we use an
interpolation argument; see the Cook--Mertz tree-evaluation procedure
in~\cite{CookMertz2024}:
\begin{equation}
  \frac{1}{(\ell+1)!}
  \sum_{S \subseteq \{0,\ldots,\ell\}}(-1)^{\ell+1-|S|}
  \left(\tau_U + \sum_{i \in S}x^{b_i\delta^i}\right)^{\ell+1}
  = \prod_{i=0}^{\ell}x^{b_i\delta^i}.
  \label{eq:finite-difference-product}
\end{equation}
We evaluate the sum in \eqref{eq:finite-difference-product} recursively in the
output register $O$ for $i = 0, \dots, \ell$. Let $P_i$ denote the program
that adds to $O$ the partial sum
\begin{equation*}
  \frac{1}{(\ell+1)!}
  \sum_{S \subseteq \{0,\ldots,i\}} (-1)^{\ell+1-|S|}
  \left(U + \sum_{j \in S} x^{b_j\delta^j}\right)^{\ell+1}.
\end{equation*}
For the base case, let $P_{-1}$ be the call-free program that adds
$(-1)^{\ell+1}U^{\ell+1}/(\ell+1)!$ to $O$.

Splitting the subsets $S \subseteq \{0, \dots, i\}$ according to whether they
contain $i$ gives
$P_i(U) = P_{i-1}(U) - P_{i-1}(U + x^{b_i\delta^i})$. Thus, we construct
$P_i$ from $P_{i-1}$ using the program
\begin{equation*}
  P_{i-1}(U)\,; C_i\,; -P_{i-1}(U)\,; C_i^{-1}.
\end{equation*}
By induction, $P_{i-1}$ restores $U$ and all the work registers. Since
$C_i^{-1}$ cancels $C_i$, the program $P_i$ also restores $U$ and all the work
registers. For $i = \ell$, the program $P_\ell$ adds
$\prod_{i=0}^{\ell} x^{b_i\delta^i}$ to $O$.

Since $P_i$ executes $P_{i-1}$ twice and each of $C_i$ and $C_i^{-1}$ once,
the number $T_i$ of input accesses satisfies
\begin{equation*}
  T_i \le 2T_{i-1} + 2q_i \le 2T_{i-1} + 2\delta \cdot 4^i, \qquad T_{-1} = 0.
\end{equation*}
It remains to bound the total number of input accesses. The recurrence gives
\begin{equation*}
  T_\ell \le 2\delta \sum_{i=0}^{\ell} 2^{\ell-i} 4^i
         = 2\delta \cdot 4^\ell \sum_{j=0}^{\ell} 2^{-j}
         < 4\delta \cdot 4^\ell
         \le 4\delta N^{\log_\delta 4},
\end{equation*}
which concludes the proof.
\end{proof}

\subsection{Generalizing to Matrices}
\label{sec:matrix-generalization}

We now apply the same four transformations to a matrix input. Fix $N=2r-1$
and let $A\in M_m(K)$. Choose a basis $e_1,\ldots,e_b$ of $M_m(K)$, where
$b=m^2$, and write $A=\sum_{i=1}^{b}a_i e_i$.

We introduce the commuting variables $z_1,\ldots,z_b$ and associate with $A$
the linear polynomial
\begin{equation*}
  \ell_A(z)=\sum_{i=1}^{b}a_i z_i.
\end{equation*}

As in the scalar construction, we encode the work registers as the
coefficients of a single monic polynomial. Let
$|\alpha|=\alpha_1+\cdots+\alpha_b$. For every $1\leq j\leq r-1$ and every
$\alpha\in\mathbb N^b$ with $|\alpha|=j$, introduce a register
$y_{j,\alpha}$. There are $\binom{b+j-1}{j}$ registers for each $j$. Set
\begin{equation}
  \mathcal P_0(u,z)
  =
  u^{r-1}+
  \sum_{j=1}^{r-1}
  \sum_{\substack{\alpha\in\mathbb N^b\\|\alpha|=j}}
    y_{j,\alpha}z^\alpha u^{r-1-j}.
  \label{eq:matrix-polynomial-registers}
\end{equation}

We now repeat the calculation of
\Cref{lem:strang-identity}, replacing the scalar input $x$ by
$\ell_A(z)$. The four successive polynomial states are
\begin{align*}
  \mathcal P_1(u,z)
  &=
  \mathcal P_0\bigl(u+\ell_A(z),z\bigr)
  +
  \frac{
    \bigl(u+\ell_A(z)\bigr)^{r-1}-u^{r-1}
  }{2},\\
  \mathcal P_2(u,z)
  &=
  \mathcal P_1\bigl(u-\ell_A(z),z\bigr)\\
  &=
  \mathcal P_0(u,z)
  +
  \frac{
    u^{r-1}-\bigl(u-\ell_A(z)\bigr)^{r-1}
  }{2},\\
  \mathcal P_3(u,z)
  &=
  \mathcal P_2\bigl(u+\ell_A(z),z\bigr)
  +
  \frac{
    \bigl(u+\ell_A(z)\bigr)^{r-1}-u^{r-1}
  }{2}\\
  &=
  \mathcal P_0\bigl(u+\ell_A(z),z\bigr)
  +
  \bigl(u+\ell_A(z)\bigr)^{r-1}-u^{r-1}.
\end{align*}
The last expression is exactly
\begin{equation*}
  \exp\bigl(\ell_A(z)D\bigr)\mathcal P_0(u,z).
\end{equation*}
The fourth transformation therefore gives
\begin{equation*}
  \mathcal P_4(u,z)
  =
  \exp\bigl(-\ell_A(z)D\bigr)\mathcal P_3(u,z)
  =
  \mathcal P_0(u,z).
\end{equation*}
Moreover, each of the four transformations uses one access to $\ell_A(z)$.
Because $\ell_A(z)$ is a linear form, this access can be implemented using a
single input access to $A$.

In the same way, substituting $\ell_A(z)$ for $x$ in the identity from
\Cref{lem:transgression} gives
\begin{equation*}
  \mathcal T(\Phi_r,\Psi_r)=\ell_A(z)^N.
\end{equation*}
Expanding the right-hand side yields
\begin{equation}
  \ell_A(z)^N
  =\left(\sum_{i=1}^{b}a_i z_i\right)^N
  =\sum_{|\alpha|=N} \binom{N}{\alpha}a^\alpha z^\alpha,
  \label{eq:formal-matrix-output}
\end{equation}
where
\begin{equation*}
  a^\alpha=a_1^{\alpha_1}\cdots a_b^{\alpha_b},
  \qquad
  \binom{N}{\alpha}=\frac{N!}{\alpha_1!\cdots\alpha_b!}.
\end{equation*}

We now convert \eqref{eq:formal-matrix-output} into $A^N$.
Expanding the matrix product gives
\begin{align*}
  A^N
  &=
  \left(\sum_{i=1}^{b}a_i e_i\right)^N\\
  &=
  \sum_{i_1,\ldots,i_N=1}^{b}
    a_{i_1}\cdots a_{i_N}
    e_{i_1}\cdots e_{i_N}\\
  &=
  \sum_{|\alpha|=N}a^\alpha M_\alpha,
\end{align*}
where
\begin{equation}
  M_\alpha
  =
  \sum_{\substack{
    i_1,\ldots,i_N\in\{1,\ldots,b\}\\
    \#\{k:i_k=s\}=\alpha_s\text{ for every }s
  }}
  e_{i_1}\cdots e_{i_N}.
  \label{eq:matrix-monomial-coefficient}
\end{equation}

For a homogeneous polynomial
\begin{equation*}
  Q(z)=\sum_{|\alpha|=N}q_\alpha z^\alpha,
\end{equation*}
define the linear map
\begin{equation}
  \mu_N(Q)
  =
  \sum_{|\alpha|=N}
    \frac{q_\alpha}{\binom{N}{\alpha}}M_\alpha.
  \label{eq:matrix-evaluation-map}
\end{equation}
Applying this map to \eqref{eq:formal-matrix-output} gives
\begin{align*}
  \mu_N\bigl(\ell_A(z)^N\bigr)
  &=
  \sum_{|\alpha|=N}
    \frac{\binom{N}{\alpha}a^\alpha}
         {\binom{N}{\alpha}}
    M_\alpha\\
  &=
  \sum_{|\alpha|=N}a^\alpha M_\alpha\\
  &=A^N.
\end{align*}
Since $\mu_N$ is linear, it can be applied separately to each call-free
output update.

The work registers require
\begin{align*}
  \sum_{j=1}^{r-1}\binom{b+j-1}{j}
  &=
  \binom{b+r-1}{r-1}-1
\end{align*}
scalar cells. The matrix output requires an additional $b=m^2$ cells. We have
therefore proved the following result.

\begin{corollary}
\label[corollary]{cor:catalytic-matrix-powering}
Let $K$ be a field of characteristic zero or of characteristic $p>2r-1$.
For every
$A\in M_m(K)$, the matrix $A^{2r-1}$ can be computed using four input accesses to
$A$ and
\begin{equation*}
  m^2+\binom{m^2+r-1}{r-1}-1
\end{equation*}
scalar registers.
\end{corollary}

\begin{remark}
As observed above, the lifted derivation program computes $\ell_A(z)^N$ using
four accesses to $\ell_A(z)$. Under the analogous degree-one map $\mu_1$,
defined by $\mu_1(z_i)=e_i$, each of these accesses corresponds exactly to an
access to $A$, since
\begin{equation*}
  \mu_1\bigl(\ell_A(z)\bigr)
  = \mu_1\left(\sum_{i=1}^{b}a_i z_i\right)
  = \sum_{i=1}^{b}a_i\mu_1(z_i)
  = \sum_{i=1}^{b}a_i e_i
  = A.
\end{equation*}
\end{remark}

\section{Catalytic Applications}
\label{sec:catalytic-applications}
We conclude with some direct applications of our register programs.

\subsection{Frequency Moments for Catalytic Streaming Algorithms}
\label{sec:catalytic-streaming}

A streaming algorithm reads a sequence
$\boldsymbol\sigma=(x_1,\ldots,x_m)\in[n]^m$ from left to right on each pass.
The catalytic variant was introduced independently by Becker et
al.~\cite{BeckerEtAl2026} and Kaplan et al.~\cite{KaplanEtAl2026}.
\begin{definition}[Catalytic streaming algorithm]
Let $\boldsymbol{\sigma}=(x_1,\ldots,x_m)$ be a stream. A
\emph{catalytic streaming algorithm} is a streaming algorithm executed by a
catalytic machine. It makes $q$ passes over $\boldsymbol{\sigma}$ if it
reads the sequence
\begin{equation*}
  x_1,x_2,\ldots,x_m
\end{equation*}
in this order $q$ times.
\end{definition}

Let $p$ be a prime, and let $\mathcal A$ be a commutative
$\mathbb F_p$-algebra of dimension $d$.
An element $y\in\mathcal A$ is \emph{streamable} if a single pass implements
$R \leftarrow R + y$ on any $\mathcal A$-valued register.
\begin{theorem}
\label{thm:streamable-commutative-powering}
Fix $r\geq2$ and $N\geq1$, and let
\begin{equation*}
  \delta=2r-1,
  \qquad
  \epsilon_r=\log_\delta4.
\end{equation*}
Suppose that an access to a streamable element $y\in\mathcal A$ uses $s$ bits
of regular memory. Then $y^N$ can be computed with:
\begin{itemize}
  \item
  at most $4\delta N^{\epsilon_r}$ passes;
  \item
  $O\left(
    rd\log\left(\delta+\left\lfloor\log_\delta N\right\rfloor\right)
  \right)$ bits of catalytic memory;
  \item
  $O_r\left(
    s+d\log\left(\delta+\left\lfloor\log_\delta N\right\rfloor\right)
    +\log N
  \right)$ bits of
  regular memory.
\end{itemize}
\end{theorem}

\begin{proof}
Let $ \ell=\left\lfloor\log_\delta N\right\rfloor$ and $M=\max\{\delta,\ell+1\}$.

We can choose a prime p such that  and $p=O(\delta+\ell)$. We work over
$\mathbb F_p$.
We want to choose $p>\max\{\delta,\ell+1\}$ to ensure that the register program
from \Cref{cor:four-call-sequential-powering} is well defined over
$\mathbb F_p$, and hence over the commutative $\mathbb F_p$-algebra
$\mathcal A$. Let us choose a  prime p such that $p=O(\delta+\ell)$.

The register program from \Cref{cor:four-call-sequential-powering} computes
$y^N$ using $r$ work registers, one passive output register, and at most
$4\delta N^{\epsilon_r}$ accesses to $y$.

Since $y$ is streamable, each access is implemented by one
pass over the stream.

An element of $\mathcal A$ can be represented using at most
$d\left\lceil\log_2p\right\rceil$ bits. Since the output is passive, it can be
kept in regular memory, while the $r$ work registers are stored in catalytic
memory.

We can then apply the reduction from register programs to catalytic machines
recalled in \Cref{lem:register-program-catalytic-simulation}. The work
registers require $O\left(
    rd\left\lceil\log_2p\right\rceil
  \right)$ bits of catalytic memory.

Regarding the regular memory, for fixed $r$, the register program has length
polynomial in $N$, so the reduction uses $O_r(\log N)$ additional regular
bits. Including the output register and the $s$ bits used to implement a
streamable access, the working space is $O_r\left(
    s+d\left\lceil\log_2p\right\rceil+\log N
  \right)$ bits.
\end{proof}

For a stream
$\boldsymbol{\sigma}=(x_1,\ldots,x_m)\in[n]^m$, define the frequency vector
$\mathbf f=(f_1,\ldots,f_n)$ and the frequency moment of order $D$ by
\begin{equation*}
  f_i=\left|\{j:x_j=i\}\right|,
  \qquad
  F_D(\boldsymbol{\sigma})=\sum_{i=1}^n f_i^D.
\end{equation*}
The
frequency vector $\mathbf f$ is streamable in one pass. Therefore,
\Cref{thm:streamable-commutative-powering} computes
$\mathbf f^D=(f_1^D,\dots,f_n^D)$. We can set a single output register on which we output all the $f_i^D$.

 Choosing to work over $\mathbb F_p$ where $p>m^D$ ensures that $p>F_D(\boldsymbol\sigma)$, and therefore,  the value computed over
$\mathbb F_p$ is exactly the frequency moment.
\begin{corollary}
\label{cor:catalytic-frequency-moments}
Let $D\geq1$ and $r\geq3$, and let
\begin{equation*}
  \epsilon_r=\log_{2r-1}4<1.
\end{equation*}
$F_D(\boldsymbol{\sigma})$ can be computed using
\begin{itemize}
  \item
  at most $4(2r-1)D^{\epsilon_r}$ passes;
  \item
  $O_r\left(nD\log(m+1)\right)$ bits of catalytic memory;
  \item
  $O_r\left(
    D\log(m+1)+\log n+\log D
  \right)$ bits of regular
  memory.
\end{itemize}
\end{corollary}

The previously known algorithms use either four passes and
$O(nD^2\log m)$ bits of catalytic memory, or $D+1$ passes and
$O(nD\log m)$ bits of catalytic memory
\cite[Thm.~2.1]{KaplanEtAl2026} and~\cite[Thm.~15]{BeckerEtAl2026},
respectively.

Our construction slightly improves the best known space bounds for computing
the frequency moment of order $D$ in the catalytic streaming model
from~\cite{BeckerEtAl2026}. At the same time, it reduces the number of passes
from $D+1$ to $O_\varepsilon(D^\varepsilon)$ for every $\varepsilon>0$.
\begin{center}
\centering
\small
\setlength{\tabcolsep}{5pt}
\begin{tabular}{lccc}
\hline
 & passes & workspace & catalytic memory \\ \hline
Kaplan et al.~\cite[Thm.~2.1]{KaplanEtAl2026}
  & $4$
  & $O(D\log m)$
  & $O(nD^2\log m)$ \\
Becker et al.~\cite[Thm.~15]{BeckerEtAl2026}
  & $D+1$
  & $O(D\log(mD))$
  & $O(nD\log(mD))$ \\
This work
  & $4(2r-1)D^{\epsilon_r}$
  & $O_r(D\log(m+1))$
  & $O_r(nD\log(m+1))$ \\
\hline
\end{tabular}
\end{center}

\subsection{Matrix Powering}
\label{sec:catalytic-matrix-powering}

We now apply the matrix-powering program from
\Cref{cor:catalytic-matrix-powering} over $\mathbb F_p$.

\begin{theorem}
\label{thm:catalytic-matrix-powering}
Let $r\geq2$ and $N\geq1$, and set
\begin{equation*}
  \delta=2r-1,
  \qquad
  \epsilon_r=\log_\delta4,
  \qquad
  \ell=\left\lfloor\log_\delta N\right\rfloor.
\end{equation*}
Let $p$ be a prime satisfying
\begin{equation*}
  p>\max\{\delta,\ell+1\}.
\end{equation*}
For every $A\in M_m(\mathbb F_p)$, the matrix $A^N$ can be computed with
at most
\begin{equation*}
  4\delta N^{\epsilon_r}
\end{equation*}
recursive calls to the input matrix $A$ and
\begin{equation}
  2m^2+\binom{m^2+r-1}{r-1}-1
  \label{eq:catalytic-matrix-powering-memory}
\end{equation}
registers over $\mathbb F_p$.
\end{theorem}

\begin{proof}
By \Cref{cor:catalytic-matrix-powering}, there is a clean passive register
program that computes $A^\delta$ using
$\binom{m^2+r-1}{r-1}-1$ work registers and an output register $U$.

We can apply the powering construction from
\Cref{cor:four-call-sequential-powering} because the powers of $A$ commute, so
the finite-difference identity in \Cref{eq:finite-difference-product} remains
valid in the matrix algebra. The condition $p>\ell+1$ ensures that
$(\ell+1)!$ is invertible in $\mathbb F_p$, while $p>\delta$ ensures that the
degree-$\delta$ four-call program is defined over $\mathbb F_p$.

The recursive-call bound follows from
\Cref{cor:four-call-sequential-powering}. The construction uses the work
registers above, the matrix-valued register $U$, and a matrix-valued passive
output register, for a total of
\begin{equation*}
  2m^2+\binom{m^2+r-1}{r-1}-1
\end{equation*}
scalar registers.
\end{proof}

In comparison with the construction of
Alekseev et al.~\cite[Thm.~2]{AlekseevEtAl2025}, this register program removes
the logarithmic factor in the recursive-call bound, reduces the number of
registers, and works directly over the smaller field $\mathbb F_p$.

\begin{center}
\centering
\small
\setlength{\tabcolsep}{6pt}
\begin{tabular}{lcc}
\hline
 & Alekseev et al.~\cite[Thm.~2]{AlekseevEtAl2025} & This work \\ \hline
recursive calls
  & $O_\varepsilon(N^\varepsilon\log N)$
  & $O_\varepsilon(N^\varepsilon)$ \\
registers
  & $O_\varepsilon\!\left(m^{2^{3/\varepsilon+1}}\log N\right)$
  & $O_\varepsilon\!\left(m^{4^{1/\varepsilon}-1}\right)$ \\
coefficient field
  & $\mathbb F_{O\left((mp)^{2^{3/\varepsilon}}\right)}$
  & $\mathbb F_p$ \\
\hline
\end{tabular}
\end{center}

\section{Conclusion and Future Directions}

We introduce derivations and their exponential operators in order to present a compact representation of register programs. This representation allows us to prove lower bounds on the number of input accesses needed to compute a degree-$d$ polynomial over fields of characteristic $0$.

On the other hand, we use this representation to find a uniform family of register programs that compute $x^{2r-1}$ with only $r$ registers over every field $K$ of characteristic $0$ or characteristic $p>2r-1$.

We use the efficiency of composition for clean register programs to obtain a program computing $x^N$ for all $N$ with $r+1$ registers and $O_r\left(N^{\log_{2r-1}4}\right)$ input accesses, which improves the bounds for computing frequency moments in catalytic streaming and for matrix powering.

The optimality of the family of register programs for powering remains unknown. Should we expect $D^{\pass}_{r,K}(4)>2r-1$? Would additional input accesses drastically increase $D^{\pass}_{r,K}$?

Beyond these specific questions, the operator framework provides a new way to search for register programs. This formulation can be adapted to many algebraic representations and leaves many possible choices of derivations to explore.

For Tree Evaluation, one may look for a logspace-uniform family such that, for every polynomial $P(x_1,\ldots,x_s)$ over a fixed finite ring $R$,
\begin{equation*}
P
\in
\RP_R^{\pass}
\bigl(s^{O(1)},O(s)\bigr),
\end{equation*}
with the additional requirement that all variables be accessed simultaneously $O(1)$ times. Such a family would imply
\begin{equation*}
\TreeEval
\in
\CSPACE
\bigl(O(\log n),\log^{O(1)}n\bigr),
\end{equation*}
improving the current subpolynomial catalytic-space bound~\cite{HenzingerPyneRagavan2026}. Finding a stronger bound of the form
\begin{equation*}
P\in\RP_R^{\pass}
\bigl(O(s),O(s)\bigr)
\end{equation*}
would prove $\TreeEval\in\L$.

We have also restricted ourselves to studying register programs over commutative rings. Studying operator-based or alternative methods for noncommutative register programs may improve results relating $\NC$ and $\CL$. So far, only
\begin{equation*}
\NC(\log n \log\log n)\subseteq \CL
\end{equation*}
is known~\cite{AlekseevEtAl2025}. Constructing a logspace-uniform family satisfying
\begin{equation*}
X^k
\in
\RP_{M_m(K)}^{\pass}
\bigl(n^{O(1)},O(1)\bigr),
\end{equation*}
for all $m,k\leq2^{O(h(n))}$, where $h(n)=\Omega(\log\log n)$, would improve this bound.

\paragraph*{Acknowledgements.}
The author would like to thank Yuval Filmus and Alexander Smal for extensive discussions and valuable insights during the early stages of this project.

\paragraph{AI Disclosure.}
The author used ChatGPT 5.6 to assist in generalizing the lower-bound methods developed by the author for $1$ and $2$ input accesses in \Cref{sec:lower-bounds}. The AI tool suggested reformulating these arguments using the operator and derivation framework presented in \Cref{sec:prelim-operators,sec:prelim-one-call}. Although mathematically equivalent to the original arguments, this formulation made their structure clearer and more insightful, and it facilitated the extension to $3$ input accesses.

The author suggested that ChatGPT 5.6 checks whether this derivation-based representation could be used to construct a uniform family of register programs. This resulted in the four-call construction computing $x^{2t-1}$ with $t$ registers, presented in \Cref{sec::uniform-four-call-program} and Appendix~B.

ChatGPT 5.6 was also used to assist with the writing of this paper. The author verified the correctness and originality of all content, including the references.
\bibliographystyle{alpha}
\bibliography{bibliography}
\newpage
\section*{Appendix A: proof of \Cref{lem:transgression}}
\label{app:transgression}

\begin{proof}
Put $N=2a+1$ and write
\begin{align}
  \Phi
  &=\phi_0h^N
    +\sum_{i=2}^{a}\phi_i c_i h^{N-i}
    +\sum_{2\le i\le j\le a}\phi_{ij}c_ic_jh^{N-i-j},
  \label{eq:phi-coordinates}\\
  \Psi
  &=\psi_0h^N
    +\sum_{i=2}^{a}\psi_i c_i h^{N-i}
    +\sum_{2\le i\le j\le a}\psi_{ij}c_ic_jh^{N-i-j}.
  \label{eq:psi-coordinates}
\end{align}
For a monomial in $h,c_2,\ldots,c_a$, define
\begin{equation*}
  \operatorname{wt}
  \left(
    h^d\prod_{j=2}^{a}c_j^{m_j}
  \right)
  =
  d+\sum_{j=2}^{a}jm_j.
\end{equation*}
Thus
\begin{equation*}
  \operatorname{wt}(h)=1,
  \qquad
  \operatorname{wt}(c_j)=j,
\end{equation*}

and every
monomial in $\Phi$ and $\Psi$ has weight $N$.

\subsection*{Reduction to a single polynomial}

In centered coordinates, the two derivations satisfy
\begin{align}
  B(h)&=1,
  &B(c_j)&=0,
  \label{eq:centered-B-action}\\
  A(h)&=-3,
  &A(c_j)&=
    a\binom{a-1}{j-1}h^{j-1}-(a-j+1)c_{j-1},
  \label{eq:centered-A-action}
\end{align}
where $c_1=0$. Since
\begin{equation*}
  [x]\exp(x\delta)(G)=\delta(G),
\end{equation*}
the coefficient of $x$ in the output increment is
\begin{equation*}
  [x]\mathcal T_a(\Phi,\Psi)=A\Phi+B\Psi.
\end{equation*}
Every solution therefore satisfies
\begin{equation}
  A\Phi+B\Psi=0.
  \label{eq:infinitesimal-coboundary-equation}
\end{equation}

Let $\mathcal W_a$ be obtained from $\mathcal V_a$ by decreasing every
power of $h$ by one:
\begin{align*}
  \mathcal W_a={}&
  \operatorname{span}\{h^{N-1}\}
  \oplus
  \operatorname{span}\{c_i h^{N-i-1}:2\le i\le a\}
  \\
  &\oplus
  \operatorname{span}
  \{c_ic_jh^{N-i-j-1}:2\le i\le j\le a\}.
\end{align*}
If $m$ is one of $1$, $c_i$, or $c_ic_j$, then
\begin{align*}
  A(mh^d)
  &=-3d\,mh^{d-1}
    +\text{terms containing fewer $c_i$ or a smaller index},\\
  B(mh^d)
  &=d\,mh^{d-1}.
\end{align*}
Since $d\ge1$, both maps from $\mathcal V_a$ to $\mathcal W_a$ are
triangular with nonzero diagonal entries. Hence they are invertible, and
\eqref{eq:infinitesimal-coboundary-equation} is equivalent to
\begin{equation}
  H=B\Psi=-A\Phi,
  \qquad
  \Psi=B^{-1}H,
  \qquad
  \Phi=-A^{-1}H.
  \label{eq:H-parametrization}
\end{equation}

For this polynomial, define
\begin{align}
  \mathcal K_x(H)={}&
  -\int_0^{x/2}\exp(tA)H\,dt
  +\int_0^x\exp(xA/2)\exp(tB)H\,dt
  \nonumber\\
  &-\int_0^{x/2}
    \exp(xA/2)\exp(xB)\exp(tA)H\,dt.
  \label{eq:K-operator}
\end{align}
The integration is coefficientwise. Using $A\Phi=-H$ and $B\Psi=H$ in
the three output differences gives
\begin{equation}
  \mathcal T_a(\Phi,\Psi)=\mathcal K_x(H).
  \label{eq:T-from-H}
\end{equation}
It remains to determine the polynomials $H$ for which $\mathcal K_x(H)$
contains neither $h$ nor any $c_i$.

\subsection*{Quadratic terms}

Put
\begin{equation*}
  R(v)=C(v)-v^a
  =\sum_{p=0}^{a-2}c_{a-p}v^p.
\end{equation*}
From
\eqref{eq:centered-B-action} and
\eqref{eq:centered-A-action}, the parts of $A$ and $B$ that preserve the number of coefficients $c_i$
are
\begin{align}
  A_0
  &=-3\frac{\partial}{\partial h}
    -\sum_{j=2}^{a}(a-j+1)c_{j-1}
      \frac{\partial}{\partial c_j},
  &
  B_0
  &=\frac{\partial}{\partial h}.
  \label{eq:A0-B0-definition}
\end{align}
Thus
\begin{align}
  A_0(h)&=-3,
  &A_0(R(v))&=-R'(v),\\
  B_0(h)&=1,
  &B_0(R(v))&=0.
  \label{eq:A0-B0-actions}
\end{align}
The coefficients of $A_0$ are independent of $h$, so
$A_0B_0=B_0A_0$. Both derivations decrease the weight by one.

Let $H^{(2)}$ be the part of $H$ containing two coefficients $c_i$.
All other terms of $A$ decrease their number, so the quadratic part of
$\mathcal K_x(H)$ is obtained by replacing $A$ and $B$ with $A_0$ and
$B_0$. Put
\begin{equation*}
  z=\frac{xA_0}{2},
  \qquad
  w=xB_0.
\end{equation*}
Since $A_0$ and $B_0$ commute, \eqref{eq:K-operator} gives
\begin{equation*}
  \mathcal K_x(H)\big|_{\deg_{\mathbf c}=2}
  =xk(z,w)H^{(2)},
\end{equation*}
where
\begin{equation}
  k(z,w)
  =-\frac{e^z-1}{2z}
   +e^z\frac{e^w-1}{w}
   -e^{z+w}\frac{e^z-1}{2z}.
  \label{eq:quadratic-loop-factor}
\end{equation}
The quotients denote their power-series expansions. Direct substitution
and expansion at the origin give
\begin{align*}
  k(z,-2z)&=0,\\
  k(z,-z)&=0,\\
  k(z,w)
  &=-\frac1{12}(2z+w)(z+w)+O_3(z,w),
\end{align*}
where $O_3(z,w)$ contains only terms of total degree at least $3$.
Consequently,
\begin{equation}
  k(z,w)=(2z+w)(z+w)u(z,w),
  \qquad
  u(0,0)=-\frac1{12}.
  \label{eq:quadratic-loop-factorization}
\end{equation}
Moreover,
\begin{equation*}
  u(xA_0/2,xB_0)
  =-\frac1{12}I
   +\sum_{r+s\ge1}
     u_{rs}\frac{x^{r+s}}{2^r}A_0^rB_0^s.
\end{equation*}
The sum decreases the weight and is therefore nilpotent on the finite
space under consideration. Hence this operator is invertible. Since
\begin{equation*}
  2z+w=x(A_0+B_0),
  \qquad
  z+w=\frac{x}{2}(A_0+2B_0),
\end{equation*}
the quadratic terms vanish if and only if
\begin{equation}
  (A_0+B_0)(A_0+2B_0)H^{(2)}=0.
  \label{eq:quadratic-kernel-equation}
\end{equation}

We now solve this equation. Define
\begin{equation*}
  J(v)=R(v-h/2).
\end{equation*}
Applying the chain rule to \eqref{eq:A0-B0-actions} gives
\begin{equation*}
  (A_0+B_0)J(v)=0.
\end{equation*}
Furthermore,
\begin{equation*}
  [v^r]J(v)
  =c_{a-r}
   +\sum_{p=r+1}^{a-2}
     \binom pr\left(-\frac h2\right)^{p-r}c_{a-p}.
\end{equation*}
For $0\le r\le a-2$, put
\begin{equation*}
  J_r=[v^r]J(v).
\end{equation*}
Since
\begin{equation*}
  J_r
  =
  c_{a-r}
  +
  \sum_{p=r+1}^{a-2}
    \binom pr
    \left(-\frac h2\right)^{p-r}c_{a-p},
\end{equation*}
the variables $c_2,\ldots,c_a$ can be recovered successively from
$h,J_0,\ldots,J_{a-2}$. Hence every polynomial $F$ has a unique
expression
\begin{equation*}
  F
  =
  \widetilde F(h,J_0,\ldots,J_{a-2}).
\end{equation*}
Moreover,
\begin{equation*}
  (A_0+B_0)(h)=-2,
  \qquad
  (A_0+B_0)(J_r)=0.
\end{equation*}
Therefore
\begin{equation*}
  (A_0+B_0)(F)
  =
  -2\frac{\partial\widetilde F}{\partial h}.
\end{equation*}
It follows that
\begin{equation*}
  F\in\ker(A_0+B_0)
  \quad\Longleftrightarrow\quad
  \widetilde F
  \text{ is independent of }h,
\end{equation*}
which means that $F$ is purely a polynomial in the coefficients of $J(v)$.

Now $(A_0+2B_0)H^{(2)}$ has degree two in the $c_i$ and weight $2a-1$,
whereas
\begin{equation*}
  \operatorname{wt}\bigl([v^r]J(v)\bigr)=a-r.
\end{equation*}
Thus \eqref{eq:quadratic-kernel-equation} implies, for some $\gamma$,
\begin{equation}
  (A_0+2B_0)H^{(2)}
  =-\gamma J(0)J'(0).
  \label{eq:first-quadratic-solution}
\end{equation}
On the other hand,
\begin{align*}
  (A_0+2B_0)J(v)&=-\frac12J'(v),\\
  (A_0+2B_0)J(0)^2&=-J(0)J'(0).
\end{align*}
It follows that
\begin{equation*}
  (A_0+2B_0)
  \bigl(H^{(2)}-\gamma J(0)^2\bigr)=0.
\end{equation*}
Applying the same triangular argument to
\begin{equation*}
  L(v)=R(v-h),
  \qquad
  (A_0+2B_0)L(v)=0,
\end{equation*}
shows that a quadratic polynomial of weight $2a$ in this kernel is a
multiple of $L(0)^2$. Therefore
\begin{equation}
  H^{(2)}
  =\gamma R(-h/2)^2+\beta R(-h)^2.
  \label{eq:two-quadratic-solutions}
\end{equation}

\subsection*{Linear terms}

Write
\begin{equation}
  H^{(1)}
  =\sum_{p=0}^{a-2}\lambda_p c_{a-p}h^{a+p}.
  \label{eq:linear-part-H-coordinates}
\end{equation}
We first use the coefficient of $x^3$ to determine the $\lambda_p$.
Expanding the three integrals in \eqref{eq:K-operator} to this order gives
\begin{equation}
  [x^3]\mathcal K_x
  =-\frac1{24}\left(A^2+3BA+2B^2\right).
  \label{eq:K-x3-operator}
\end{equation}
Indeed, the coefficients of $x$ and $x^2$ cancel, while the coefficient
of $x^3$ is
\begin{align*}
  &-\frac1{48}A^2
  +\frac18A^2+\frac14AB+\frac16B^2\\
  &\qquad
  -\frac7{48}A^2-\frac14AB-\frac18BA-\frac14B^2
  =-\frac1{24}\left(A^2+3BA+2B^2\right).
\end{align*}

For $0\le q\le a-2$, applying this operator to
\eqref{eq:linear-part-H-coordinates} gives
\begin{align}
  &\left[c_{a-q}h^{a+q-2}x^3\right]
  \mathcal K_x(H^{(1)})
  \nonumber\\
  &\qquad=
  -\frac{(a+q)(a+q-1)}{12}\lambda_q
  -\frac{q(a+q-1)}8\lambda_{q-1}
  -\frac{q(q-1)}{24}\lambda_{q-2},
  \label{eq:linear-x3-coefficient}
\end{align}
where $\lambda_{-1}=\lambda_{-2}=0$. Only these three coefficients occur
because each of the operators $A^2$, $BA$, and $B^2$ can decrease the
index of a coefficient $c_i$ by at most two.

To compute the same coefficient for $H^{(2)}$, put
\begin{equation*}
  U_\theta=R(-\theta h),
  \qquad
  \theta\in\{1/2,1\}.
\end{equation*}
Equations \eqref{eq:centered-B-action} and
\eqref{eq:centered-A-action} give
\begin{align*}
  A(U_\theta)
  &=(3\theta-1)R'(-\theta h)+P_\theta,
  &B(U_\theta)&=-\theta R'(-\theta h),\\
  P_\theta
  &=ah^{a-1}
    \left((1-\theta)^{a-1}-(-\theta)^{a-1}\right).
\end{align*}
The part independent of the $c_i$ in $A(R'(-\theta h))$ is
\begin{equation*}
  Q_\theta
  =a(a-1)h^{a-2}
   \left((1-\theta)^{a-2}-(-\theta)^{a-2}\right).
\end{equation*}
Keeping only the terms containing one coefficient $c_i$ gives
\begin{align*}
  A^2(U_\theta^2)\big|_{\deg_{\mathbf c}=1}
  ={}&4(3\theta-1)P_\theta R'(-\theta h)
      +2(3\theta-1)U_\theta Q_\theta
      -6U_\theta\frac{\partial P_\theta}{\partial h},\\
  BA(U_\theta^2)\big|_{\deg_{\mathbf c}=1}
  ={}&-2\theta P_\theta R'(-\theta h)
      +2U_\theta\frac{\partial P_\theta}{\partial h},\\
  B^2(U_\theta^2)\big|_{\deg_{\mathbf c}=1}
  ={}&0.
\end{align*}
The derivatives of $P_\theta$ cancel after these expressions are
substituted into \eqref{eq:K-x3-operator}. Hence
\begin{align}
  [x^3]\mathcal K_x(U_\theta^2)
  \big|_{\deg_{\mathbf c}=1}
  =-\frac1{12}\Bigl(&
    (3\theta-2)P_\theta R'(-\theta h)
    \nonumber\\
    &+(3\theta-1)U_\theta Q_\theta
  \Bigr).
  \label{eq:quadratic-x3-theta}
\end{align}
Since
\begin{equation*}
  [c_{a-q}]U_\theta=(-\theta h)^q,
  \qquad
  [c_{a-q}]R'(-\theta h)=q(-\theta h)^{q-1},
\end{equation*}
we obtain
\begin{align}
  &\left[c_{a-q}h^{a+q-2}x^3\right]
  \mathcal K_x(H^{(2)})
  \nonumber\\
  &\qquad=
  -\gamma\frac{a}{3\cdot2^a}
  \left((1-\epsilon_a)(a-1)+\epsilon_a q\right)
  \left(-\frac12\right)^q
  \nonumber\\
  &\qquad\quad+
  \beta\frac{a(-1)^{a+q}}{12}(q+2a-2),
  \label{eq:quadratic-x3-coefficient}
\end{align}
where $\epsilon_a=0$ for odd $a$ and $\epsilon_a=1$ for even $a$.

Adding \eqref{eq:linear-x3-coefficient} and
\eqref{eq:quadratic-x3-coefficient} gives the triangular equations
\begin{align}
  0={}&
  -\gamma\frac{a}{3\cdot2^a}
  \left((1-\epsilon_a)(a-1)+\epsilon_a q\right)
  \left(-\frac12\right)^q
  \nonumber\\
  &+\beta\frac{a(-1)^{a+q}}{12}(q+2a-2)
  -\frac{(a+q)(a+q-1)}{12}\lambda_q
  \nonumber\\
  &-\frac{q(a+q-1)}8\lambda_{q-1}
  -\frac{q(q-1)}{24}\lambda_{q-2}.
  \label{eq:triangular-linear-equations}
\end{align}
Define
\begin{equation}
  s_{a,q}
  =\sum_{j=0}^{q}2^{q-j}
    \binom{a+j-1}{j-\epsilon_a},
  \qquad
  s_{a,-1}=0.
  \label{eq:saq-definition}
\end{equation}
The identity
\begin{equation*}
  s_{a,q}-2s_{a,q-1}
  =\binom{a+q-1}{q-\epsilon_a}
\end{equation*}
reduces \eqref{eq:triangular-linear-equations} to
\begin{equation}
  \lambda_q
  =-\gamma\,2^{2-a}
   \left(-\frac12\right)^q
   \frac{s_{a,q}}{\binom{a+q}{q}}
   +2(-1)^{a+q}\beta.
  \label{eq:linear-coefficients-before-beta}
\end{equation}

We now use one further coefficient equation to determine $\beta$. For
this calculation, put
\begin{equation}
  S_t(v,h)
  =(v-t)^a-v^a
   +\frac{(v+h-t)^a-(v+h-3t)^a}{2}.
  \label{eq:A-source-polynomial}
\end{equation}
Equations \eqref{eq:centered-B-action} and
\eqref{eq:centered-A-action} imply
\begin{align}
  \exp(tA)(h)&=h-3t,
  \nonumber\\
  \exp(tA)(R(v))&=R(v-t)+S_t(v,h),
  \label{eq:first-transformed-R}\\
  \exp(xA/2)\exp(tB)(h)&=h-3x/2+t,
  \nonumber\\
  \exp(xA/2)\exp(tB)(R(v))
    &=R(v-x/2)+S_{x/2}(v,h),
  \label{eq:second-transformed-R}\\
  \exp(xA/2)\exp(xB)\exp(tA)(h)
    &=h-x/2-3t,
  \nonumber\\
  \exp(xA/2)\exp(xB)\exp(tA)(R(v))
    &=R(v-t-x/2)
  \nonumber\\
  &\quad+S_{x/2}(v-t,h)+S_t(v,h-x/2).
  \label{eq:third-transformed-R}
\end{align}
All ensuing integrals are evaluations of
\begin{align}
  \int_0^{cx}t^r(h-dx-et)^m\,dt
  ={}&\sum_{k=0}^{m}\sum_{\nu=0}^{k}
    (-1)^k\binom mk\binom k\nu
    d^{k-\nu}e^\nu
  \nonumber\\
  &\qquad\cdot
    \frac{c^{r+\nu+1}}{r+\nu+1}
    h^{m-k}x^{r+k+1},
  \label{eq:elementary-path-integral}
\end{align}
followed by the classical finite-difference identity
\begin{equation}
  \sum_{\nu=0}^{m}(-1)^{m-\nu}\binom m\nu
  \binom{\lambda+\nu}{k}
  =\binom{\lambda}{k-m}
  \label{eq:finite-difference-binomial-identity}
\end{equation}
from \cite[Chs.~2 and~5]{GrahamKnuthPatashnik1994}.

Suppose first that $a$ is even. Then $s_{a,0}=0$, and
\eqref{eq:linear-coefficients-before-beta} gives $\lambda_0=2\beta$.
Substituting $q=0$ and $h=0$ into
\eqref{eq:first-transformed-R}--\eqref{eq:third-transformed-R}, and then
using \eqref{eq:elementary-path-integral}, gives
\begin{align*}
  \left[c_ax^{a+1}\right]
  \mathcal K_x\bigl(\gamma R(-h/2)^2\bigr)
  &=0,\\
  \left[c_ax^{a+1}\right]
  \mathcal K_x\bigl(\beta R(-h)^2\bigr)
  &=\frac{\beta}{a+1}
  \left[
    -1-\frac43\left(\frac32\right)^{a+1}
    +\frac43\left(\frac12\right)^{a+1}
    +\frac23\,2^{a+1}
  \right],\\
  \left[c_ax^{a+1}\right]
  \mathcal K_x\bigl(\lambda_0c_ah^a\bigr)
  &=\frac{\lambda_0}{3(a+1)}
  \left[
    2\left(\frac32\right)^{a+1}
    -2\left(\frac12\right)^{a+1}
    -2^{a+1}
  \right].
\end{align*}
After substituting $\lambda_0=2\beta$, their sum is
\begin{equation*}
  \left[c_ax^{a+1}\right]
  \mathcal K_x(H^{(2)}+H^{(1)})
  =-\frac{\beta}{a+1}.
\end{equation*}
For odd $a$, the same substitution with $q=1$ gives
\begin{equation*}
  \left[c_{a-1}x^{a+2}\right]
  \mathcal K_x(H^{(2)}+H^{(1)})
  =-\frac{a\beta}{2(a+1)(a+2)}.
\end{equation*}
These coefficients must vanish, so in both cases
\begin{equation*}
  \beta=0.
\end{equation*}
Consequently,
\begin{align}
  H^{(2)}
  &=\gamma R(-h/2)^2,
  \label{eq:quadratic-part-of-H}\\
  H^{(1)}
  &=-\gamma\,2^{2-a}
    \sum_{p=0}^{a-2}
    \left(-\frac12\right)^p
    \frac{s_{a,p}}{\binom{a+p}{p}}
    c_{a-p}h^{a+p}.
  \label{eq:linear-part-of-H}
\end{align}

It remains to check the coefficients of higher powers of $x$. Substituting
\eqref{eq:quadratic-part-of-H} and \eqref{eq:linear-part-of-H} into
\eqref{eq:K-operator}, expanding with
\eqref{eq:first-transformed-R}--\eqref{eq:third-transformed-R}, and using
\eqref{eq:finite-difference-binomial-identity} gives, for
$0\le q\le a-2$ and $3\le k\le a+q+1$,
\begin{align}
  &\left[c_{a-q}h^{a+q+1-k}x^k\right]
  \mathcal K_x(H^{(2)})
  \nonumber\\
  &\qquad=
  \gamma\,2^{2-a}
  \sum_{p=0}^{q}
    \left(-\frac12\right)^p
    \frac{s_{a,p}}{\binom{a+p}{p}}
  \left[c_{a-q}h^{a+q+1-k}x^k\right]
  \mathcal K_x(c_{a-p}h^{a+p}),
  \label{eq:remaining-H2-coefficients}\\
  &\left[c_{a-q}h^{a+q+1-k}x^k\right]
  \mathcal K_x(H^{(1)})
  \nonumber\\
  &\qquad=
  -\gamma\,2^{2-a}
  \sum_{p=0}^{q}
    \left(-\frac12\right)^p
    \frac{s_{a,p}}{\binom{a+p}{p}}
  \left[c_{a-q}h^{a+q+1-k}x^k\right]
  \mathcal K_x(c_{a-p}h^{a+p}).
  \label{eq:remaining-H1-coefficients}
\end{align}
The two expressions cancel. Hence all terms containing exactly one
coefficient $c_i$ vanish.

\subsection*{Constant term and normalization}

Write
\begin{equation*}
  H=H^{(2)}+H^{(1)}+\kappa h^{2a}.
\end{equation*}
The contribution of the last term is determined by
\begin{align}
  \mathcal K_x(h^{2a})={}&
  -\int_0^{x/2}(h-3t)^{2a}\,dt
  +\int_0^x\left(h-\frac{3x}{2}+t\right)^{2a}\,dt
  \nonumber\\
  &-\int_0^{x/2}
    \left(h-\frac x2-3t\right)^{2a}\,dt.
  \label{eq:scalar-loop-polynomial}
\end{align}
Substituting \eqref{eq:quadratic-part-of-H} and
\eqref{eq:linear-part-of-H} into \eqref{eq:K-operator}, and simplifying
with \eqref{eq:elementary-path-integral} and
\eqref{eq:finite-difference-binomial-identity}, gives
\begin{align}
  \mathcal K_x(H)\big|_{c_2=\cdots=c_a=0}
  ={}&
  \left(
    \kappa+\frac{2(-1)^a}{\binom{2a}{a}}\gamma
  \right)\mathcal K_x(h^{2a})
  \nonumber\\
  &+\gamma(-1)^a
    \frac{(a!)^2}{(2a+1)!}x^{2a+1}.
  \label{eq:constant-term-reduction}
\end{align}
The first term depends on $h$, since
\begin{equation*}
  \left[h^{2a-2}x^3\right]\mathcal K_x(h^{2a})
  =-\frac{(2a)(2a-1)}{12}\ne0.
\end{equation*}
It vanishes precisely when
\begin{equation}
  \kappa=-\frac{2(-1)^a}{\binom{2a}{a}}\gamma.
  \label{eq:constant-part-of-H}
\end{equation}
For this value,
\begin{equation}
  \mathcal K_x(H)
  =\gamma(-1)^a
   \frac{(a!)^2}{(2a+1)!}x^{2a+1}.
  \label{eq:final-output-before-normalization}
\end{equation}

We have determined $H^{(2)}$, $H^{(1)}$, and the constant part of $H$
from the single parameter $\gamma$. Choosing
\begin{equation}
  \gamma
  =(-1)^a\frac{(2a+1)!}{(a!)^2}
  =(-1)^a(2a+1)\binom{2a}{a}
  \label{eq:gamma-normalization}
\end{equation}
gives
\begin{equation*}
  \mathcal K_x(H)=x^{2a+1}.
\end{equation*}
Finally, define $\Phi_a$ and $\Psi_a$ from this $H$ by
\eqref{eq:H-parametrization}. Equation \eqref{eq:T-from-H} gives
\begin{equation*}
  \mathcal T_a(\Phi_a,\Psi_a)=x^{2a+1}.
\end{equation*}
Each step above is reversible: the quadratic equations determine
$H^{(2)}$ up to $\gamma$ and $\beta$, the linear equations force
$\beta=0$ and determine $H^{(1)}$, the constant equation determines
$\kappa$, and the final normalization determines $\gamma$. The solution
is therefore unique.

\paragraph{Logspace uniformity and positive characteristic.}
The register program is logspace-uniform. Indeed,
\Cref{eq:saq-definition,eq:quadratic-part-of-H,eq:linear-part-of-H,eq:constant-part-of-H,eq:gamma-normalization}
express $H$ using a polynomial number of binomial terms.

Moreover, we do not need to compute the inversions in
\Cref{eq:H-parametrization}. By \Cref{eq:H-parametrization}, we have
$B\Psi=H$ and $A\Phi=-H$.

\Cref{eq:centered-B-action} gives $B=\partial_h$. Since every monomial in
$\mathcal V_a$ is divisible by $h$, we have $\Psi(0,\mathbf c)=0$, and
therefore
\begin{equation*}
  \Psi(h,\mathbf c)=\int_0^h H(s,\mathbf c)\,ds.
\end{equation*}
Similarly, for $\Phi$, \Cref{eq:first-transformed-R} gives
$\exp(tA)(h)=h-3t$. Since $\Phi\in\mathcal V_a$ is also divisible by $h$ and
$\exp$ is a morphism, letting $t=\frac{h}{3}$ gives
\begin{equation*}
  \bigl(\exp((h/3)A)\Phi\bigr)(h,\mathbf c)=0.
\end{equation*}
On the other hand, using $A\Phi=-H$, we have
\begin{equation*}
  \frac{d}{dt}\bigl(\exp(tA)\Phi\bigr)
  =\exp(tA)(A\Phi)
  =-\exp(tA)H.
\end{equation*}
Integrating this identity from $0$ to $h/3$ yields
\begin{equation*}
  \Phi(h,\mathbf c)
  =\int_0^{h/3}
    \bigl(\exp(tA)H\bigr)(h,\mathbf c)\,dt.
\end{equation*}
Thus both $\Phi$ and $\Psi$ are obtained by integrating
polynomials.
Finally, all divisions appearing in these formulas, in the centered
change of coordinates, and in the triangular one-access conjugations have
denominators whose prime factors are at most $N=2r-1$. Consequently, all
program coefficients belong to $\mathbb Z[1/(2r-1)!]$. For every prime
$p>2r-1$, reduction modulo $p$ is therefore well defined and preserves
all polynomial identities proving restoration and correctness. The same
logspace-uniform four-call program thus works over every field of
characteristic $p>2r-1$.
\end{proof}

\end{document}